\let\latexaddtocontents\addtocontents
\documentclass[a4paper,twocolumn,11pt, unpublished]{quantumarticle}
\let\addtocontents\latexaddtocontents

\pdfoutput=1

\usepackage[utf8]{inputenc}

\usepackage{amsmath}
\usepackage{fullpage}
\usepackage{amsthm}
\usepackage{amsfonts}
\usepackage{amssymb}
\usepackage{todonotes}
\usepackage{mathtools}
\usepackage{braket}
\usepackage{algorithm}
\usepackage{algpseudocode}
\usepackage{appendix}
\usepackage{xcolor}
\usepackage{fontawesome}
\usepackage{pgfplots}
\pgfplotsset{compat=1.3}
\usepackage{tikz}
\usetikzlibrary{quantikz}
\usetikzlibrary{tikzmark}
\usetikzlibrary{decorations.pathreplacing}
\usepackage{subcaption}
\usepackage{hyperref}  
\usepackage[capitalize]{cleveref}  
\usepackage[style=numeric-comp,
backend=bibtex8,
doi=true,
isbn=false,
url=false,
eprint=false,
maxbibnames=5,
sorting=none]{biblatex}
\renewbibmacro{in:}{}
\crefdefaultlabelformat{#2#1#3}  
\crefformat{equation}{Eq.~#2#1#3}  

\definecolor{ForestGreen}{RGB}{34,139,34}
\definecolor{DarkRed}{RGB}{204,0,0}

\newcommand{\QAE}{\textrm{QAE}}
\newcommand{\QPE}{\textrm{QPE}}
\newcommand{\QAA}{\textrm{QAA}}
\newcommand{\algoAE}{\textsc{EigenvalueEstimation}}

\renewcommand{\tilde}{\widetilde}  
\renewcommand{\epsilon}{\varepsilon}  
\DeclareMathOperator{\Tr}{Tr}

\theoremstyle{plain}
\newtheorem{theorem}{Theorem}
\newtheorem{problem}[theorem]{Problem}
\newtheorem{lemma}[theorem]{Lemma}
\newtheorem{corollary}[theorem]{Corollary}

\begin{document}

\title{A square-root speedup for finding the smallest eigenvalue}
\author{Alex Kerzner}
    \email{alexkerzner1@gmail.com}
    \affiliation{softwareQ Inc., Waterloo, Ontario, N2L 0C7, Canada}
    \author{Vlad Gheorghiu}
    \affiliation{softwareQ Inc., Waterloo, Ontario, N2L 0C7, Canada}
    \affiliation{Institute for Quantum Computing, University of Waterloo, Waterloo, Ontario, N2L 3G1, Canada}
\author{Michele Mosca}
    \affiliation{softwareQ Inc., Waterloo, Ontario, N2L 0C7, Canada}
    \affiliation{Institute for Quantum Computing, University of Waterloo, Waterloo, Ontario, N2L 3G1, Canada}
    \affiliation{Perimeter Institute for Theoretical Physics, Waterloo, Ontario, N2L 2Y5, Canada}
\author{Thomas Guilbaud}
    \affiliation{Laboratory for Reactor Physics and Systems Behaviour, EPFL, Lausanne, Vaud, Switzerland}
    \affiliation{Transmutex SA, Vernier, Geneva, Switzerland}
\author{Federico Carminati}
    \affiliation{Transmutex SA, Vernier, Geneva, Switzerland}
\author{Fabio Fracas}
    \affiliation{Transmutex SA, Vernier, Geneva, Switzerland}
    \affiliation{University of Padua, Padua, Italy}
\author{Luca Dellantonio}
    \email{l.dellantonio@exeter.ac.uk}
    \affiliation{Institute for Quantum Computing, University of Waterloo, Waterloo, Ontario, N2L 3G1, Canada}
    \affiliation{Department of Physics and Astronomy, University of Exeter, Stocker Road, Exeter EX4 4QL, United Kingdom}

\begin{abstract}
We describe a quantum algorithm for finding the smallest eigenvalue of a Hermitian matrix. This algorithm combines Quantum Phase Estimation and Quantum Amplitude Estimation to achieve a quadratic speedup with respect to the best classical algorithm in terms of matrix dimensionality, i.e., $\widetilde{\mathcal{O}}(\sqrt{N}/\epsilon)$\footnote{
In this work $\tilde{O}$ ignores terms that are polylogarithmic in $N$ or $1/\epsilon$.
} black-box queries to an oracle encoding the matrix, where $N$ is the matrix dimension and $\epsilon$ is the desired precision. In contrast, the best classical algorithm for the same task requires $\Omega(N)\text{polylog}(1/\epsilon)$ queries. In addition, this algorithm allows the user to select any constant success probability. We also provide a similar algorithm with the same runtime that allows us to prepare a quantum state lying mostly in the matrix's low-energy subspace. We implement simulations of both algorithms
and demonstrate their application to problems in quantum chemistry and materials science.
\end{abstract}

\maketitle

\section{Introduction}

Finding the smallest eigenvalue of a Hermitian operator is a ubiquitous task in quantum computing, computer science, and the natural sciences more generally \cite{aspuru2005simulated,nielsen2002quantum}. The special case of $k$-local Hamiltonians with $k=O(1)$, which correspond to several quantum systems arising in nature, has been extensively studied in the quantum computing literature \cite{gharibian2015quantum,Haase2021resourceefficient,Paulson2021,Shlosberg2023adaptiveestimation,1311.3161}. Local Hamiltonians can also encode optimization problems in the sense that the ground state and ground state energy correspond to the minimizer and minimum value of a cost function of interest \cite{ShenVyalyiKitaev2003}. The power of ground state energy estimation for local Hamiltonians is reflected in the fact that a decision version of this problem is \textsc{QMA}-complete \cite{kempe2006complexity}. Although this means that quantum computers are unlikely to provide an exponential speedup for this task, we may hope for polynomial speedups, which can be crucial in practice. Moreover, we might ask whether such speedups can be retained when considering nonlocal Hamiltonians.

One well-known quantum algorithm for solving this problem is the quantum phase estimation (QPE) algorithm \cite{kitaev1995quantum}, which, given an eigenstate of a unitary operator, estimates its eigenvalue. In general, however, it is not clear whether the ground state can be efficiently prepared. One approach \cite{aspuru2005simulated} is to first prepare the ground state of a known Hamiltonian, then perform an adiabatic evolution to slowly evolve the initial state into the ground state of the target Hamiltonian. However, the time required for this adiabatic evolution is inversely proportional to the gap between the smallest eigenvalues of the Hamiltonian \cite{albash2018adiabatic}, and in general, this gap may vanish. Moreover, for general Hamiltonians, this gap is usually unknown a priori, making it impossible to assess the validity of the obtained result. 

Another approach consists in approximately preparing the ground state by solving a simpler related problem first. In Ref.~\cite{jaksch2003eigenvector}, the authors consider the case where the Hamiltonian corresponds to a discretization of a partial differential operator, and an approximation of the ground state can be efficiently prepared by first considering that same PDE on a coarse grid. However, the simpler related problem may not always exist, and in general, the solution to this related problem may not be sufficiently similar to the solution to the actual problem.

Other quantum approaches include variational algorithms \cite{cerezo2021variational,chan2023hybrid,Ferguson2021,Bespalova2021}. These prepare a circuit $C(\vec{\theta})$, where $\vec{\theta}$ is a set of parameters with which $C$ varies continuously. The smallest eigenvalue of an $n$-qubit Hamiltonian $H$ is (hopefully) the minimum value of $\braket{0^n|C(\vec{\theta})^\dagger H C(\vec{\theta})|0^n}$ ($\ket{0^n}\equiv\ket{0}^{\otimes n}$), and so by computing this value and repeatedly varying $\vec{\theta}$ according to a classical minimization strategy, we might hope to find the global minimum. Variational algorithms, however, have few performance guarantees, as their cost landscapes may suffer from barren plateaus and local minima \cite{Uvarov2020}.

In this paper, we describe a quantum algorithm that has (1) a guaranteed speedup over the best possible classical algorithm with respect to the dimension of the considered matrix; and (2) a success probability that the user can select. This algorithm makes use only of two basic subroutines: QPE and quantum amplitude estimation (QAE) \cite{Brassard_2002} and requires $\widetilde{O}(\sqrt{N}/\epsilon)$ black-box queries to the matrix, where $N$ is the matrix dimension and $\epsilon$ the desired precision. For fixed $\epsilon$, this is an improvement over classical algorithms, which require $\Omega(N)\text{polylog}(1/\epsilon)$ queries. 

Let $\lambda_0 \leq \ldots \leq \lambda_{N-1}$ be the eigenvalues of our Hermitian matrix and $\ket{\psi_0},\ldots,\ket{\psi_{N-1}}$ the associated eigenvectors. A simplified version of the algorithm is shown in \Cref{fig:main_figure}. We define a subroutine $\mathcal{A}$ (\Cref{fig:main_figure} (iii)), which first prepares $\frac{1}{\sqrt{N}}\sum_j \ket{\psi_j}\ket{\bar{\psi_j}}$ ($\bar{\cdot}$ denotes the element-wise complex conjugate) and then uses QPE to approximately prepare $\frac{1}{\sqrt{N}}\sum_j\ket{\lambda_j}\ket{\psi_j}\ket{\bar{\psi_j}}$, where $\lambda_j$ is encoded in the clock register employed for QPE. We then repeatedly use {\QAE} to infer $\lambda_0$ with a binary search (\Cref{fig:binary_search} (i)) in a manner similar to D\"{u}rr and H{\o}yer's quantum algorithm for minimum-finding \cite{durr1996quantum}. Once we have estimated $\lambda_0$, it is possible to use quantum amplitude amplification (QAA) \cite{Brassard_2002} to amplify the low-energy eigenstates and prepare a state lying mostly in the low-energy eigenspace.

\begin{figure*}

\centering
\includegraphics[width=\textwidth]{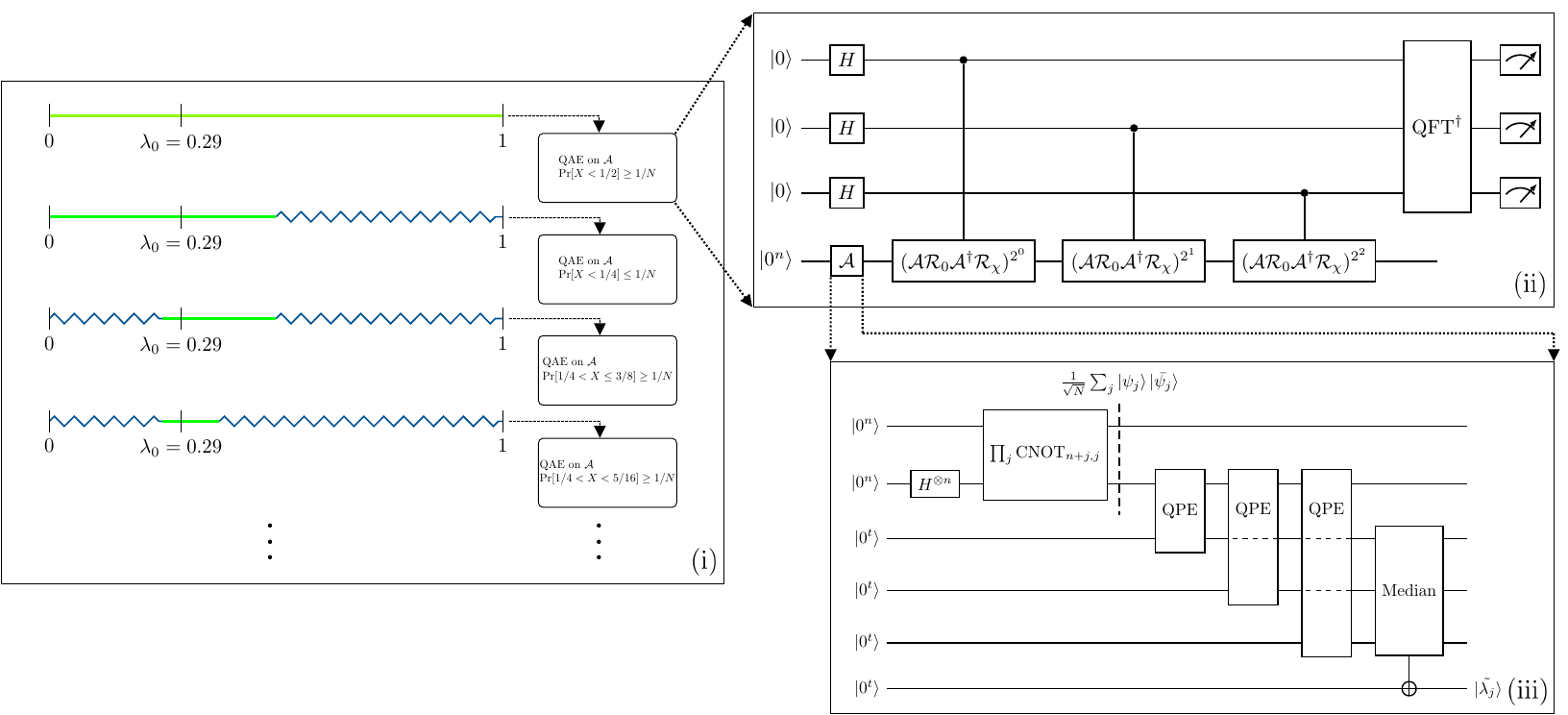}
\caption{A depiction of a simplified version of our algorithm when $\lambda_0 = 0.29$. (i) An example of the first few steps of the binary search portion of our algorithm. At each step, the candidate interval, i.e., the interval in which $\lambda_0$ may lie (green), is cut in half (zigzag blue). Here, $X\in[0,1)$ is the random variable obtained by measuring the first register of $\mathcal{A}\ket{0}$ and considering the outcome $\ket{X_1 X_2\ldots}$ as a number $X \coloneqq 0.X_1 X_2\ldots$. (ii) The Quantum Amplitude Estimation (QAE) portion of algorithm. Here, $\ket{\bar{\psi_j}}$ represents the state $\ket{\psi_j}$ with all amplitudes replaced by their complex conjugates, and the $\mathcal{R}$ operators are the usual Grover reflection operators \cite{nielsen2002quantum}. (iii) The algorithm $\mathcal{A}$, including the median trick.}
\label{fig:binary_search}

\label{fig:main_figure}
\end{figure*}

The notation and terminology employed in this work is the following. An $N\times N$ matrix is \textit{sparse} if it has $O(\log N)$ nonzero entries in each row and column, $||H||_\text{max}$ denotes the largest absolute value of any entry of the matrix $H$, and $||H|| = \sqrt{\lambda_\mathrm{max}(H^\dagger H)}$ is the spectral norm $\max_j|\lambda_j(H)|$ in the case where $H$ is Hermitian. We consider a query complexity model in which we have black-box access to an oracle that, on input $(i,k)$, outputs the value and column number of the $k$-th nonzero entry in row $i$ \cite{berry2014exponential}. Finally, in scenarios where a precision $\epsilon$ has been specified, $\Pi$ denotes an operator that projects onto the subspace spanned by $\{\ket{\psi_j} : \lambda_j - \lambda_0 < \epsilon\}$.

We formally consider the following problem.

\begin{problem}\label{prob:eig}
The input is an error tolerance $\epsilon > 0$, oracle access to a sparse $N\times N$ Hermitian matrix $H$ with eigenvalues in $(\epsilon,1-\epsilon)$ and $||H||_\text{max} = O(1)$, and a constant $\gamma < 1$ independent of $\epsilon$ and $N$. The output is an estimate of $\lambda_0$ up to additive error $\epsilon$, as well as a mixed state $\rho$ such that $\Tr(\Pi \rho) \geq \gamma$.
\end{problem}
Without loss of generality, we set $\gamma=2/3$ in the remainder of this work, and remark that for other values of $\gamma$ the arguments are virtually the same. The algorithm is captured by the following theorem.
\begin{theorem}\label{thm:main}
    There exists a quantum algorithm that solves \Cref{prob:eig} using $\widetilde{O}(\sqrt{N}/\epsilon)$ oracle queries and one- and two-qubit gates.
\end{theorem}

In addition to these theoretical results, we have implemented both algorithms using \cite{Qiskit}. In \Cref{app:implementation} we discuss some technical details regarding this implementation, however the full code may be found at Ref.~\cite{Kerzner_EigenvalueFinding_2023}.

This work is structured as follows. In \Cref{sec:alg}, we describe the solution of \Cref{prob:eig}, and in \Cref{sec:applications}, we discuss and simulate the application of our algorithms to two concrete problems. \Cref{app:extension_general} discusses an extension to general sparse Hermitian matrices, and \Cref{app:details} provides some technical details omitted in \Cref{sec:alg}.

\section{Algorithm description}\label{sec:alg}
In this section, we outline our proof of \Cref{thm:main}, with some details deferred to \Cref{app:details}. We first describe the subroutine $\mathcal{A}$ in \Cref{sec:subroutineA}, and then demonstrate how it is employed in the {\algoAE} algorithm in \Cref{sec:usingA}, which provides the desired estimate of $\lambda_0$. Finally, we show how this estimate can be used in conjunction with {\QAA} to prepare $\rho$ as in \Cref{prob:eig}. From now on, we set $n = \log(N)$ and $m=\log(1/\epsilon)$, where the logarithms are in base $2$.

\subsection{The subroutine $\mathcal{A}$}\label{sec:subroutineA}

The subroutine $\mathcal{A}$ comprises two steps. The first is state preparation.

For our algorithm to work, it suffices to have an initial state whose overlap with the ground state $\ket{\psi_0}$ is roughly $1/N$. We expect that a Haar-random state would suffice for this purpose \cite{wootters1990random, kuperberg19random}, but for clarity we instead prepare $\frac{1}{\sqrt{N}}\sum_j \ket{j}\ket{j}$, which requires $n$ Hadamard and $\mathrm{CNOT}$ gates (see \Cref{fig:main_figure} (iii)). This state may be rewritten as $\frac{1}{\sqrt{N}}\sum_j\ket{\psi_j}\ket{\bar{\psi_j}}$, with $\bar{\cdot}$ denoting the complex conjugate. We will not use the second register during the remainder of $\mathcal{A}$. This register cannot, however, be traced out, because the {\QAE} algorithm which will be used later requires the application of $\mathcal{A}^\dagger$.

The second step is to apply QPE on the unitary $e^{2\pi iH}$, making use of Ref.~\cite{berry2014exponential} to perform the necessary controlled Hamiltonian evolution. The complexity of approximately applying $e^{2\pi i H}$ according to the routine in Ref.~\cite{berry2014exponential} is sublogarithmic in the inverse error between the approximate and exact unitaries. Therefore, to simplify our algorithm's analysis, we assume that this error is negligible compared to the other ones considered below. We remark that the approximation of $e^{2\pi i H}$ is efficient provided $H$ is sparse, which is why we require sparsity of the matrix $H$ (see \Cref{prob:eig}).

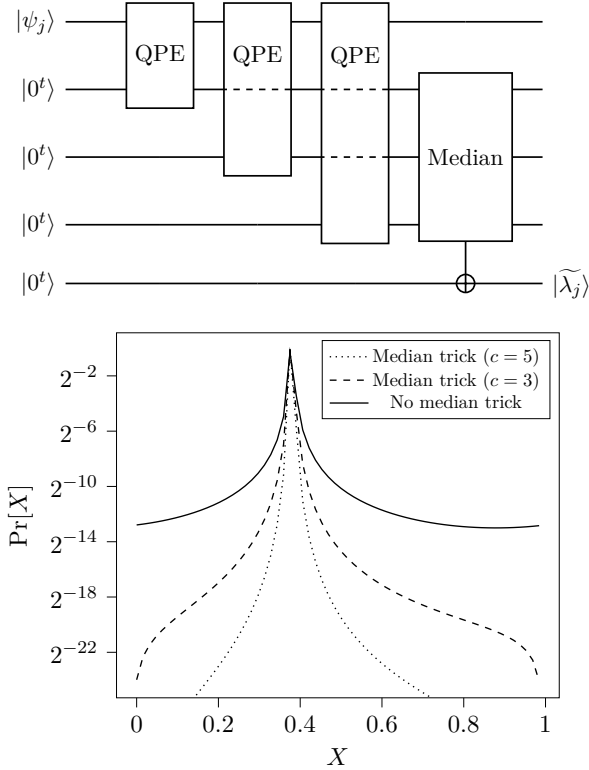
\begin{figure}
\centering
\begin{subfigure}{0.9\textwidth}
\begin{tikzpicture}
\node[scale=0.8]{
\begin{quantikz}[execute at end picture={
            \draw[dashed, thick] (\tikzcdmatrixname-2-4.east) -- 
            (\tikzcdmatrixname-2-4.west);
            \draw[dashed, thick] (\tikzcdmatrixname-2-5.east) -- 
            (\tikzcdmatrixname-2-5.west);
            \draw[dashed, thick] (\tikzcdmatrixname-3-5.east) -- 
            (\tikzcdmatrixname-3-5.west);
}
]
\lstick{$\ket{\psi_j}$}        & \qw & \gate[wires=2]{\QPE} & \gate[wires=3, label style={yshift=0.55cm}]{\QPE} & \gate[wires=4, label style={yshift=1.1cm}]{\QPE} & \qw & \qw \\ 
\lstick{$\ket{0^t}$}        & \qw &  &  &  & \gate[wires=3]{\textrm{Median}}  & \qw\\
\lstick{$\ket{0^t}$}        & \qw & \qw & & & &  \qw\\
\lstick{$\ket{0^t}$}        & \qw & \qw &\qw &  &  & \qw\\
\lstick{$\ket{0^t}$}        & \qw & \qw & \qw & \qw & \targ{}\vqw{-1}& \qw \rstick{$\ket{\Tilde{\lambda_j}}$}
\end{quantikz}
};
\end{tikzpicture}
\end{subfigure}
\par\medskip
\begin{subfigure}{\textwidth}
\begin{tikzpicture}[scale = 0.85]
\begin{axis}[
legend style = {nodes={scale=0.75, transform shape}},
title style = {align=center},
log basis y={2},
tick align=outside,
tick pos=left,
x grid style={darkgray176},
xlabel={\(\displaystyle X\)},
xmin=-0.04921875, xmax=1.03359375,
xtick style={color=black},
y grid style={darkgray176},
ylabel={\(\displaystyle \Pr[X]\)},
ymin=2.43271466162285e-08, ymax=2.16856774707114,
ymode=log,
ytick style={color=black},
ytick={9.31322574615479e-10,1.49011611938477e-08,2.38418579101562e-07,3.814697265625e-06,6.103515625e-05,0.0009765625,0.015625,0.25,4,64},
yticklabels={
  \(\displaystyle {2^{-30}}\),
  \(\displaystyle {2^{-26}}\),
  \(\displaystyle {2^{-22}}\),
  \(\displaystyle {2^{-18}}\),
  \(\displaystyle {2^{-14}}\),
  \(\displaystyle {2^{-10}}\),
  \(\displaystyle {2^{-6}}\),
  \(\displaystyle {2^{-2}}\),
  \(\displaystyle {2^{2}}\),
  \(\displaystyle {2^{6}}\)
}
]
\addplot[semithick, dotted]
table {%
0.0 2.8381112317996056e-11
0.015625 2.135733240510352e-10
0.03125 6.339469317006904e-10
0.046875 1.3667027406118895e-09
0.0625 2.5242677749525187e-09
0.078125 4.271235424279765e-09
0.09375 6.8510165848569004e-09
0.109375 1.0628601978779592e-08
0.125 1.6160744251430603e-08
0.140625 2.4314137566755372e-08
0.15625 3.6470295572560225e-08
0.171875 5.4892673780801574e-08
0.1875 8.340987441285279e-08
0.203125 1.2874374690063971e-07
0.21875 2.0322571363584928e-07
0.234375 3.3069584624015447e-07
0.25 5.602693167297639e-07
0.265625 1.0014156567943387e-06
0.28125 1.923763188415353e-06
0.296875 4.084639556388565e-06
0.3125 1.0033722414764257e-05
0.328125 3.099883287101166e-05
0.34375 0.00014443339724789757
0.359375 0.0017480881862928736
0.375 0.9802691700940998
0.390625 0.017204843941262468
0.40625 0.00046644399871142674
0.421875 7.556256711665989e-05
0.4375 2.22445958782077e-05
0.453125 8.851918083171288e-06
0.46875 4.227272850412352e-06
0.484375 2.2814396039630693e-06
0.5 1.3441206784198887e-06
0.515625 8.457911676471105e-07
0.53125 5.601554136467623e-07
0.546875 3.8643406467800247e-07
0.5625 2.7558964651767036e-07
0.578125 2.020087013865848e-07
0.59375 1.51512656627391e-07
0.609375 1.1586282846856327e-07
0.625 9.007104596100714e-08
0.640625 7.10086114976997e-08
0.65625 5.665201515872137e-08
0.671875 4.565687801632472e-08
0.6875 3.710855383705635e-08
0.703125 3.037147700259255e-08
0.71875 2.4995582948007242e-08
0.734375 2.0656712102925526e-08
0.75 1.711779718601114e-08
0.765625 1.4203071794823876e-08
0.78125 1.1780618780716227e-08
0.796875 9.750369518846321e-09
0.8125 8.03573431277821e-09
0.828125 6.5776964948961396e-09
0.84375 5.330609483904675e-09
0.859375 4.2591950997830186e-09
0.875 3.3364102543174867e-09
0.890625 2.541962265656221e-09
0.90625 1.8613318802221356e-09
0.921875 1.2852214632762735e-09
0.9375 8.093933636032431e-10
0.953125 4.349077478718348e-10
0.96875 1.688174151896278e-10
0.984375 2.5437607629673147e-11
};
\addplot [semithick, dashed]
table {%
0 6.01397144556538e-08
0.015625 1.90847844974445e-07
0.03125 3.40786204674041e-07
0.046875 5.15812760443665e-07
0.0625 7.23561298622581e-07
0.078125 9.74154485851578e-07
0.09375 1.28124931236407e-06
0.109375 1.66360126770301e-06
0.125 2.14745630464719e-06
0.140625 2.77029729245204e-06
0.15625 3.58687014837174e-06
0.171875 4.67917090787091e-06
0.1875 6.17356854733878e-06
0.203125 8.27132759608334e-06
0.21875 1.1305531839753e-05
0.234375 1.5853039711031e-05
0.25 2.29691068787652e-05
0.265625 3.471860258909e-05
0.28125 5.55005317137675e-05
0.296875 9.5787188420576e-05
0.3125 0.000184597771458506
0.328125 0.000422361259304233
0.34375 0.00130903307830468
0.359375 0.0080821314137052
0.375 0.943638472070905
0.390625 0.0415174571571635
0.40625 0.00301374788831031
0.421875 0.000778083968537323
0.4375 0.000311061387179653
0.453125 0.00015559375622209
0.46875 8.92458154230982e-05
0.484375 5.61498630274932e-05
0.5 3.77656728983303e-05
0.515625 2.67157487424451e-05
0.53125 1.96608888011945e-05
0.546875 1.49364432335959e-05
0.5625 1.16476413566705e-05
0.578125 9.28354860471097e-06
0.59375 7.53754608486403e-06
0.609375 6.21781288855019e-06
0.625 5.20002687434059e-06
0.640625 4.40112096224244e-06
0.65625 3.76409978892568e-06
0.671875 3.24893104924793e-06
0.6875 2.82690486045464e-06
0.703125 2.47704132136982e-06
0.71875 2.18374452982082e-06
0.734375 1.93523563723258e-06
0.75 1.722484526097e-06
0.765625 1.53846749057725e-06
0.78125 1.37764213219858e-06
0.796875 1.2355694075024e-06
0.8125 1.10863677252666e-06
0.828125 9.93851539747452e-07
0.84375 8.88683304719242e-07
0.859375 7.90940629402647e-07
0.875 6.98671301887366e-07
0.890625 6.10078167478209e-07
0.90625 5.23444193427161e-07
0.921875 4.37061353748433e-07
0.9375 3.49158240177261e-07
0.953125 2.57821061992516e-07
0.96875 1.60901843399204e-07
0.984375 5.59060138936981e-08
};
\addplot [semithick]
table {%
0 0.000141592598539933
0.015625 0.000147680036023204
0.03125 0.000154930420040735
0.046875 0.000163534677897459
0.0625 0.000173735729633643
0.078125 0.000185844575205762
0.09375 0.000200262797672638
0.109375 0.000217514525115031
0.125 0.00023829263108521
0.140625 0.000263526863726852
0.15625 0.000294486612485067
0.171875 0.00033293995973432
0.1875 0.000381407173721035
0.203125 0.000443578580290885
0.21875 0.000525030914768051
0.234375 0.000634513174739094
0.25 0.000786384891364413
0.265625 0.00100555815764728
0.28125 0.00133837883357198
0.296875 0.00187926233934197
0.3125 0.00284578540371454
0.328125 0.00483716711362132
0.34375 0.0100478200511866
0.359375 0.0324634995273191
0.375 0.810610160468473
0.390625 0.0901039754854258
0.40625 0.0165829842916334
0.421875 0.00673976006088434
0.4375 0.00364349832188734
0.453125 0.00228648117252821
0.46875 0.0015736141497092
0.484375 0.00115349163381608
0.5 0.000885360081094
0.515625 0.000703925918256335
0.53125 0.000575542997437946
0.546875 0.000481443018461146
0.5625 0.000410486859206704
0.578125 0.000355725090487745
0.59375 0.000312640559256415
0.609375 0.000278194615750844
0.625 0.00025028287421139
0.640625 0.000227411217990832
0.65625 0.000208495790805344
0.671875 0.00019273558552191
0.6875 0.000179529021314813
0.703125 0.000168417989019282
0.71875 0.000159049513717131
0.734375 0.000151148990908626
0.75 0.000144501193392782
0.765625 0.000138936601352408
0.78125 0.000134321448569253
0.796875 0.000130550410842253
0.8125 0.00012754120820488
0.828125 0.000125230621147385
0.84375 0.000123571575536661
0.859375 0.000122531057790347
0.875 0.000122088697916608
0.890625 0.000122235914252541
0.90625 0.000122975557567966
0.921875 0.000124322028999323
0.9375 0.000126301880248686
0.953125 0.000128954939383479
0.96875 0.000132336045309348
0.984375 0.00013651752324297
};
\legend{Median trick ($c=5$), Median trick ($c=3$), No median trick};
\end{axis}

\end{tikzpicture}
\end{subfigure}
\caption{\textit{Top:} Circuit for the median trick with $c=3$,
performing the unitary transformation $\ket{w,x,y,z}\mapsto\ket{w,x,y,z+\mathrm{Median}(w,x,y)}$. 
\textit{Bottom:} QPE on the unitary matrix $\mathrm{diag}(1, e^{2\pi i \lambda})$ with input state $\ket{1}$, $t=6$ bits of precision, and $\lambda = 0.3789$.
We plot the discrete probability distributions obtained without and with the median trick ($c$ in the legend). The curves cluster more closely around $\lambda$ as $c$ increases, in agreement with the theoretical result.
}
\label{fig:median_trick}
\end{figure}

For each eigenvalue $\lambda_j$, QPE does not always yield $\ket{\lambda_j}$, as the binary expansion of $\lambda_j$ may be too long for an exact representation \cite{nielsen2002quantum}. Instead, QPE prepares a state with the following property. Let $X = 0.X_1 X_2\ldots \in [0,1]$ be the random variable obtained when the state is measured in the computational basis and $\ket{X_1 X_2\ldots}$ is observed. Then $\Pr[X\approx \lambda_j]$ is close to $1$.\footnote{
This statement is a slight simplification, as QPE only considers the eigenvalues modulo $1$. So if $\lambda_j\approx 0$ (resp. $\lambda_j\approx 1$), then the probability of observing $X\approx 1$ (resp. $X\approx 0$) may be high. The possibility of this type of wraparound error is precluded by the requirement $\lambda_j\in(\epsilon,1-\epsilon)$ in the statement of \Cref{prob:eig}.
} Formally, we will require below that $\Pr[|X - \lambda_j|\leq \epsilon/2] \geq 1- \delta$, where $\delta = \frac{1}{2N+2}$. (For more information on this choice of $\delta$, see \Cref{app:details}.)

One way to enforce this requirement is to use $m + O(\log(1/\delta))$ bits of precision in QPE \cite[\S 5.2.1]{nielsen2002quantum}. However, because the runtime of QPE scales exponentially in the number of bits of precision, this approach cannot be taken to have quantum advantage. Instead, we employ the so-called \textit{median trick} \cite{jerrum1986random, nagaj2009fast}, which we depict in \Cref{fig:median_trick}. To do so, we first select $t = m + O(1)$, such that when $t$-bit QPE is executed, the probability of observing such an $X$ is less than $1/4$.\footnote{
Recall that when using standard QPE, in order to obtain $m$ bits of precision with probability $1 - \zeta$, it suffices to use $t = m + \lceil\log(2 + \frac{1}{2\zeta})\rceil$ ancilla qubits \cite{nielsen2002quantum}.
} We run $t$-bit QPE some number $c$ times in serial, where the $i$-th execution of QPE uses the $\ket{\psi_j}$ register as the eigenvector register and uses the $i$-th clock register as the eigenvalue register (see \Cref{fig:median_trick}). We then reversibly compute the median by XOR-ing it into a ``median register'' of size $t$. To simplify the analysis of the algorithm, we uncompute the $c$ clock registers (however in practice this is unnecessary). After the median trick, the state of the system is
\begin{equation}\label{eq:Psi}
    \ket{\Psi} = \frac{1}{\sqrt{N}} \sum_{j=0}^{N-1} \ket{0^{ct}}\ket{\tilde{\lambda_j}}\ket{\psi_j}\ket{\bar{\psi_j}}.
\end{equation}
We want $\ket{\tilde{\lambda_j}}$ such that, if measured, the observed string $X$ satisfies $|0.X_1X_2\ldots - \lambda_0| > \epsilon/2$ with probability at most $\delta$. If we view the $c$ clock registers as independent Bernoulli trials with success probability $3/4$, then we can see that the probability that such an $X$ is observed in the median register is at most the probability that $c/2$ or more of those trials fail. That probability is given by $\sum_{i=c/2}^c \binom{c}{i}\left(\frac{1}{4}\right)^i \left(\frac{3}{4}\right)^{c-i}$, and choosing $c = O(\log(1/\delta))$ forces this sum to be less than $\delta$. \Cref{fig:median_trick} illustrates how the median trick can be used to amplify the probability that $X$ is close to the true eigenvalue.

In conclusion, we define $\mathcal{A}$ to be the algorithm comprising the two steps above. It first prepares $\frac{1}{\sqrt{N}}\sum_j \ket{\psi_j}\ket{\bar{\psi_j}}$, and then it runs QPE, using the median trick to prepare the state $\ket{\Psi}$ in \Cref{eq:Psi}. 
The runtime of $\mathcal{A}$ is dominated by applying QPE with $t= O(\log(1/\epsilon))$ bits of precision. Using Ref.~\cite{berry2014exponential} we execute $t$-bit QPE using $\widetilde{O}(\log(N)/\epsilon)$ oracle queries and one- and two-qubit gates. We remark that we may also use $H$ with nonconstant entries, and our algorithm's runtime will scale linearly in $||H||_\mathrm{max}$.

\subsection{Determining $\lambda_0$ and preparing $\rho$}\label{sec:usingA}

Below, we describe how to employ $\mathcal{A}$ in our main algorithms. Since both {\QAE} and {\QAA} are required as subroutines, we first recall them. A schematic representation of {\QAE} is given in \Cref{fig:binary_search} (ii).
\begin{theorem}[Quantum Amplitude Estimation {\QAE} \cite{Brassard_2002}]\label{thm:AE}
    Let $\mathcal{B}$ be a unitary matrix and $\ket{\Phi} = \mathcal{B}\ket{0^k}$. For a function $\chi: \{0,1\}^k \rightarrow \{\mathrm{bad}, \mathrm{good}\}$, let us write $\ket{\Phi} = \sqrt{p_{\rm b}}\ket{\Phi_{\rm b}} + \sqrt{p_{\rm g}}\ket{\Phi_{\rm g}}$, where $p_{\rm b}, p_{\rm g}$ are probabilities and $\ket{\Phi_{\rm b}}$ (resp. $\ket{\Phi_{\rm g}}$) lies in the subspace spanned by $\{\ket{x} : \chi(x) = \mathrm{bad}\}$ (resp. $\{\ket{x} : \chi(x) = \mathrm{good}\}$). Then for each even positive integer $M$, there exists a quantum algorithm that outputs $\widetilde{p_{\rm g}}$ such that 
    \begin{equation}\label{eq:AEerror}
        |\widetilde{p_{\rm g}} - p_{\rm g}| < 2\pi \frac{\sqrt{p_{\rm g}(1-p_{\rm g})}}{M} + \frac{\pi^2}{M^2}
    \end{equation}
    with probability at least $8/\pi^2$. This algorithm requires $O(M)$ applications of $\mathcal{B}$ and $\mathcal{B}^\dagger$.
\end{theorem}

\begin{theorem}[Quantum Amplitude Amplification {\QAA} \cite{Brassard_2002}]\label{thm:AA}
    Using the same notation as \Cref{thm:AE}, there exists a quantum algorithm that prepares $\ket{\Phi_{\rm g}}$ with probability $\geq 1/2$, together with a flag indicating success or failure, using $O(1/\sqrt{p_{\rm g}})$ applications of $\mathcal{B}$ and $\mathcal{B}^\dagger$.
\end{theorem}

The {\QAE} algorithm is based on QPE, and therefore one can again use the median trick to boost the success probability from $8/\pi^2$ to a user-defined threshold.\footnote{
Since {\algoAE} uses $m$ applications of {\QAE}, in order for it to have constant success probability $\nu$, we require each application of {\QAE} to have a success probability $1-\Delta$ that satisfies $(1-\Delta)^m > \nu$, or equivalently, $\Delta < 1-\nu^{1/m}$. The number of clock registers needed for the median trick is therefore $O(\log(1/\Delta)) = O(\log(1/(1-\nu^{1/m}))) = O(\log(m))$, which is negligible.
} Similarly, because {\QAA} outputs a flag indicating whether the algorithm was successful, we may boost the overall success probability by simply repeating the algorithm some small number of times.

\begin{algorithm}[H]
\caption{{\algoAE}}\label{alg:minfinding}
\begin{algorithmic}[1]
\State $k \gets 537$
\State $M\gets \sqrt{\frac{kN}{1-\delta}}$
\State $q \gets \frac{1-\delta}{N}\left(\frac{1}{2} + \frac{\sqrt{2}\pi}{\sqrt{k}} + \frac{\pi^2}{k}\right)$
\State $y_0 \gets \frac{1}{2}$
\For{$i=1,\ldots m$}
\State $\widetilde{p}\gets \QAE(\mathcal{A}, \chi_{y_i}, M)$ \label{line:6}
\If{$\widetilde{p} > q$}
    \State $y_i \gets y_{i-1} - 1/2^{i+1}$
\Else
    \State $y_i \gets y_{i-1} + 1/2^{i+1}$
\EndIf
\EndFor
\State \Return $y_m$
\end{algorithmic}
\end{algorithm}

As represented in \Cref{fig:binary_search} and \Cref{alg:minfinding}, our main algorithm, {\algoAE}, is an iterative protocol that uses $\mathcal{A}$ in conjunction with {\QAE} to build an estimate $y_m \approx \lambda_0$ via a binary search \cite{cormen2022introduction}. At each step $i \in \{ 1,\ldots, m \}$ of the binary search, we determine whether the median register in \Cref{eq:Psi} encodes any state $\ket{\tilde{\lambda_j}}$ such that $\lambda_j$ is less than the current estimate $y_i$, and update the estimate accordingly. In \Cref{app:details}, we show that our estimates $y_i$ satisfy $|y_i - \lambda_0| \leq 1/2^{i+1} + \epsilon/2$. In particular, this tells us that the final estimate $y_m$ satisfies $|y_m - \lambda_0| \leq 1/2^m = \epsilon$, as desired. A numerical demonstration of convergence for a set of random Hamiltonians is given in \Cref{fig:convergence}, where it is possible to see that for each $i$ the estimation error $|y_i - \lambda_0|$ is below $1/2^{i+1} + \epsilon/2$, and is often several orders of magnitude smaller. In the next paragraph, we describe how each step $i$ of the binary search is practically carried out. For clarity we drop the subscript $i$.

\begin{figure}[htp]
\centering
\begin{tikzpicture}[scale=0.875,
]
\definecolor{darkgray176}{RGB}{176,176,176}
\definecolor{green}{RGB}{0,128,0}

\begin{axis}[
legend style = {nodes={scale=0.8, transform shape}},
log basis y={2},
tick align=outside,
tick pos=left,
x grid style={darkgray176},
xlabel={\(\displaystyle i\)},
xmin=-0.3, xmax=6.3,
xtick style={color=black},
y grid style={darkgray176},
ylabel={\(\displaystyle |y_i - \lambda_0|\)},
ymin=0.00033323157305711, ymax=0.719901004508144,
ymode=log,
ytick style={color=black},
ytick={6.103515625e-05,0.000244140625,0.0009765625,0.00390625,0.015625,0.0625,0.25,1,4},
yticklabels={
  \(\displaystyle {2^{-14}}\),
  \(\displaystyle {2^{-12}}\),
  \(\displaystyle {2^{-10}}\),
  \(\displaystyle {2^{-8}}\),
  \(\displaystyle {2^{-6}}\),
  \(\displaystyle {2^{-4}}\),
  \(\displaystyle {2^{-2}}\),
}
]
\addplot [thick, black, dotted]
table {%
0 0.5078125
1 0.2578125
2 0.1328125
3 0.0703125
4 0.0390625
5 0.0234375
6 0.015625
};
\addplot [line width=0.1pt, blue!30]
table {%
0 0.436978537223037
1 0.186978537223037
2 0.0619785372230369
3 0.000521462776963105
4 0.0307285372230369
5 0.0151035372230369
6 0.00729103722303689
};
\addplot [line width=0.1pt, blue!30]
table {%
0 0.463971195874056
1 0.213971195874056
2 0.0889711958740564
3 0.0264711958740564
4 0.00477880412594362
5 0.0108461958740564
6 0.00303369587405638
};
\addplot [line width=0.1pt, blue!30]
table {%
0 0.411573051620037
1 0.161573051620037
2 0.0365730516200365
3 0.0259269483799635
4 0.00532305162003655
5 0.0103019483799635
6 0.00248944837996345
};
\addplot [line width=0.1pt, blue!30]
table {%
0 0.464764236107915
1 0.214764236107915
2 0.0897642361079154
3 0.0272642361079154
4 0.00398576389208456
5 0.0116392361079154
6 0.00382673610791544
};
\addplot [line width=0.1pt, blue!30]
table {%
0 0.41613959578879
1 0.16613959578879
2 0.0411395957887898
3 0.0213604042112102
4 0.00988959578878977
5 0.00573540421121023
6 0.00207709578878977
};
\addplot [line width=0.1pt, blue!30]
table {%
0 0.400183571159913
1 0.150183571159913
2 0.025183571159913
3 0.037316428840087
4 0.00606642884008698
5 0.00955857115991302
6 0.00174607115991302
};
\addplot [line width=0.1pt, blue!30]
table {%
0 0.44437808345438
1 0.19437808345438
2 0.0693780834543795
3 0.00687808345437954
4 0.0243719165456205
5 0.00874691654562046
6 0.000934416545620459
};
\addplot [line width=0.1pt, blue!30]
table {%
0 0.402481373721519
1 0.152481373721519
2 0.0274813737215188
3 0.0350186262784812
4 0.00376862627848121
5 0.0118563737215188
6 0.00404387372151879
};
\addplot [line width=0.1pt, blue!30]
table {%
0 0.320784906142381
1 0.0707849061423806
2 0.0542150938576194
3 0.00828490614238059
4 0.0229650938576194
5 0.00734009385761941
6 0.000472406142380588
};
\addplot [line width=0.1pt, blue!30]
table {%
0 0.247784433656915
1 0.00221556634308473
2 0.122784433656915
3 0.0602844336569153
4 0.0290344336569153
5 0.0134094336569153
6 0.00559693365691527
};
\addplot [line width=0.1pt, blue!30]
table {%
0 0.294293795531301
1 0.0442937955313012
2 0.0807062044686988
3 0.0182062044686988
4 0.0130437955313012
5 0.00258120446869878
6 0.00523129553130122
};
\addplot [line width=0.1pt, blue!30]
table {%
0 0.471139667101889
1 0.221139667101889
2 0.0961396671018891
3 0.0336396671018891
4 0.00238966710188914
5 0.0132353328981109
6 0.00542283289811086
};
\addplot [line width=0.1pt, blue!30]
table {%
0 0.410405639996475
1 0.160405639996475
2 0.0354056399964746
3 0.0270943600035254
4 0.00415563999647457
5 0.0114693600035254
6 0.00365686000352543
};
\addplot [line width=0.1pt, blue!30]
table {%
0 0.431345119392355
1 0.181345119392355
2 0.056345119392355
3 0.00615488060764503
4 0.025095119392355
5 0.00947011939235497
6 0.00165761939235497
};
\addplot [line width=0.1pt, blue!30]
table {%
0 0.436290093599886
1 0.186290093599886
2 0.0612900935998861
3 0.00120990640011394
4 0.0300400935998861
5 0.0144150935998861
6 0.00660259359988606
};
\addplot [line width=0.1pt, blue!30]
table {%
0 0.361882098120527
1 0.111882098120527
2 0.0131179018794727
3 0.0493820981205273
4 0.0181320981205273
5 0.0025070981205273
6 0.0053054018794727
};
\addplot [line width=0.1pt, blue!30]
table {%
0 0.32670605908428
1 0.0767060590842804
2 0.0482939409157196
3 0.0142060590842804
4 0.0170439409157196
5 0.00141894091571965
6 0.00639355908428035
};
\addplot [line width=0.1pt, blue!30]
table {%
0 0.280203999774704
1 0.0302039997747036
2 0.0947960002252964
3 0.0322960002252964
4 0.00104600022529638
5 0.0145789997747036
6 0.00676649977470362
};
\addplot [line width=0.1pt, blue!30]
table {%
0 0.474876990570719
1 0.224876990570719
2 0.0998769905707186
3 0.0373769905707186
4 0.00612699057071862
5 0.00949800942928138
6 0.00168550942928138
};
\addplot [line width=0.1pt, blue!30]
table {%
0 0.418546295341497
1 0.168546295341497
2 0.0435462953414971
3 0.0189537046585029
4 0.0122962953414971
5 0.00332870465850289
6 0.00448379534149711
};
\addplot [thick, blue, dashed]
table {%
0 0.382771070722315
1 0.147066258496602
2 0.0603129472867657
3 0.0278162218484409
4 0.0142548645523323
5 0.00780387518356952
6 0.00391986296716631
};
\addplot [thick, green, dashed]
table {%
0 0.484343716250865
1 0.24560889840206
2 0.124566718138627
3 0.0620667181386275
4 0.0315234487373657
5 0.0161520099661799
6 0.00843079289685955
};
\legend{$1/2^{i+1} + \epsilon/2$, Single run,,,,,,,,,,,,,,,,,,,, Average error, Maximum error};
\end{axis}

\end{tikzpicture}
\caption{Convergence of $y_i$ to $\lambda_0$. Taking $N = 8$ and $\epsilon = 1/2^6$, we executed {\algoAE} on $1000$ diagonal matrices with eigenvalues sampled uniformly from $(\epsilon, 1-\epsilon)$. The light blue curves show the errors at step $i$ for twenty of those matrices, while the dashed curves indicate the average and maximum errors at step $i$, as indicated in the legend. The value of $1/2^{i+1}+\epsilon/2$ is given by the black dotted line. In accordance with \Cref{lem:minfindinginduction}, all errors lie below the black curve.}
\label{fig:convergence}
\end{figure}
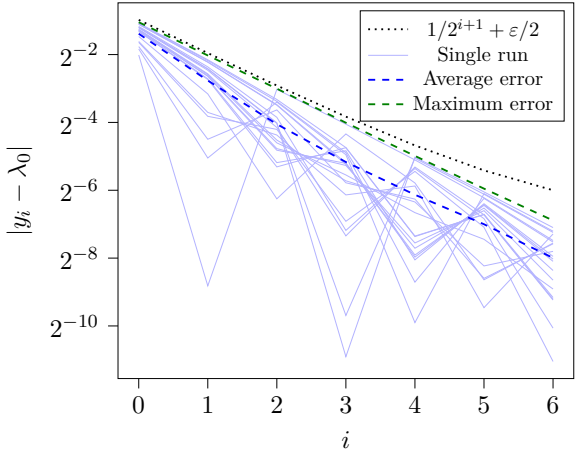

For $y \in (0,1)$, we define the function $\chi_{y}$ employed for {\QAE} (see \Cref{thm:AE}), as the step function mapping $x\in\{0,1\}^t$ to $1$ if and only if $0.x_1 x_2\ldots x_t < y$. We write $\widetilde{p} \gets \QAE(\mathcal{A}, \chi_y, M)$ to denote the output of {\QAE} with $\mathcal{B}=\mathcal{A}$, $\chi=\chi_y$, and $M$ specified in the statement of \Cref{thm:AE}. We remark that when running {\QAE} using the function $\chi_y$, that function is only to be applied to the median register of $\ket{\Psi}$ in \Cref{eq:Psi}. Therefore technically, $\chi_y$ is a function on $\{0,1\}^{2n+(c+1)t}$, but it only depends on $t$ of those bits. In \algoAE, we define a threshold $q \approx 0.71\cdot\frac{1-\delta}{N}$ such that (with one exception) $\tilde{p} > q$ if and only if $\lambda_0 < y$. Therefore, by comparing the output $\tilde{p}$ of {\QAE} to $q$, we can determine whether we must increase or decrease $y$. The other parameters $k, M$ are explained in \Cref{app:details} and ensure that {\QAE} uses $O(\sqrt{N})$ applications of $\mathcal{A},\mathcal{A}^\dagger$.

Before discussing the runtime of {\algoAE}, let us sketch our algorithm for preparing the state $\rho$ in \Cref{prob:eig}, and again defer details to \Cref{app:details}. Our overall strategy is to first produce an estimate $\theta_0$ of $\lambda_0$ using {\algoAE}, and then use {\QAA} to project $\ket{\Psi}$ (from \Cref{eq:Psi}) onto the subspace spanned by 
\begin{equation}\label{eq:good_subspace}
    \{\ket{0^{ct}}\ket{x}\ket{\psi_j}\ket{\bar{\psi_j}} : x\approx \theta_0\}.
\end{equation}
Because we have chosen $\delta$ to be sufficiently small, the states that are amplified will mostly be those with $j$ such that $\lambda_j \approx \theta_0$. Additionally, because $\theta_0 \approx \lambda_0$, and the projection of $\ket{\Psi}$ onto this subspace has length $\Omega(\sqrt{\frac{1}{N}})$, by \Cref{thm:AA}, we require only $O(\sqrt{N})$ applications of $\mathcal{A},\mathcal{A}^\dagger$.

Finally, we state the runtimes of {\algoAE} and the preparation of $\rho$. In \Cref{sec:subroutineA}, we stated that $\mathcal{A}$ requires $\tilde{O}(\text{polylog}(N)/\epsilon)$ oracle queries and one- and two-qubit gates. When running {\QAE}, we set $M=O(\sqrt{\frac{N}{1-\delta}}) = O(\sqrt{N})$, hence we require $O(\sqrt{N})$ applications of $\mathcal{A}$ and $\mathcal{A}^\dagger$. Finally, we iterate this $O(\log(1/\epsilon))$ times. The product of these costs gives us the number of queries and gates promised by \Cref{thm:main}. Now we turn to preparing $\rho$. The first stage is to run {\algoAE}, whose complexity we just discussed. In \Cref{app:details}, we show that the overlap $p_\mathrm{g}$ between $\ket{\Psi}$ and the good subspace in \Cref{eq:good_subspace} is in $\Omega(1/N)$, and therefore we again require just $O(\sqrt{N})$ applications of $\mathcal{A}$ and $\mathcal{A}^\dagger$.

\subsection{Extensions and remarks}\label{sec:extensions}
In this section we relax the requirement that eigenvalues lie in the $[0,1]$ interval and extend {\algoAE} to solve a variant of \Cref{prob:eig}. We state this result in the corollary below, with the proof given in \Cref{app:extension_general}.%
\begin{corollary}\label{cor:gen1}
    Let $H$ be a sparse Hermitian matrix, $\epsilon > 0$ an error tolerance, $\gamma <1$ a constant, and $\widetilde{||H||} \geq ||H||$ be an estimate of $||H||$. There exists a quantum algorithm using $\tilde{O}(\tilde{||H||}\sqrt{N}/\epsilon)$ oracle queries and one- and two-qubit gates that approximates $\lambda_0$ up to error $\epsilon$ and prepares a state $\rho$ with $\Tr(\Pi\rho) > \gamma$.
\end{corollary}

The runtime depends now on $\widetilde{||H||}$. In practice, precisely estimating $||H||$ quickly may be challenging, however estimates $\widetilde{||H||}$ can often be obtained by using a priori knowledge regarding $H$, and generic estimation methods like the Gershgorin Circle Theorem \cite{grasselli2008numerical}. For proving \Cref{cor:gen1}, we transform the problem to \Cref{prob:eig}. This is a matter of performing a linear transformation to the spectrum of $H$, which can be achieved by adding multiples of the identity and multiplying by a constant. However, by doing so, any error incurred on the rescaled matrix is magnified by a factor of $2\tilde{||H||}$ when we scale back to the original spectrum. This fact must be considered in cases where $||H||$ (or $\widetilde{||H||}$) is large. For instance, in the example presented in \Cref{sec:Maxwell_waves}, $||H||$ increases with $N$. In future work, it would be interesting to investigate more advanced methods for shifting and rescaling $H$ such that errors are limited. It might be possible to rescale $H$ nonlinearly.

Finally, we remark that our minimum-finding subroutine is similar to that of D\"{u}rr and H{\o}yer \cite{durr1996quantum}, who demonstrated a square-root speedup for finding the minimum of a function $T$ defined on $\{0,1\}^n$. However, in that work is assumed that it is possible to perfectly prepare the state $\ket{x}\ket{T(x)}$. This is generally infeasible in our case, due to the imprecision of QPE. Additionally, Ref.~\cite{durr1996quantum} employs a different approach to the binary search. They iterate over progressively smaller values of $T(x)$, whereas {\algoAE} iterates over $[0,1]$, finding values $y_i$ which are progressively closer to $\lambda_0$. We believe that our approach is beneficial due to the fact that it is easier to make statements about convergence with respect to the interval $[0, 1]$. In contrast, the convergence of \cite{durr1996quantum} is stated in terms of the image of $T$. In \Cref{app:extension_general}, we also explain how to generalize {\algoAE} to find the eigenvalue of $H$ that is nearest to a user-specified value (instead of simply finding $\lambda_0$).

\section{Applications}\label{sec:applications}

In this section, we showcase {\algoAE} by applying it to two different problems. In \Cref{sec:chemistry}, we determine the ground state energy of the hydrogen molecule $\text{H}_2$ for different atomic distances. Afterwards, we consider a set of partial differential equations (PDEs) in \Cref{sec:Maxwell_waves} and determine their minimum eigenvalue and corresponding eigenvector. In both cases, {\algoAE} allows for determining the desired result up to the chosen accuracy. These examples, albeit proofs of principle, demonstrate the advantages of our approach. Namely, that it deterministically finds the ground state and ground state energy (up to a chosen precision) with fixed and known time and resource complexities. 

We remark that, for specific problems, other classical and/or quantum algorithms may converge faster or with a higher accuracy to the desired result. Specifically, {\algoAE} has a quadratic advantage in $N$ but scales worse in $\epsilon$ compared to classical approaches (see \Cref{sec:Maxwell_waves_results}). On the other hand, it is challenging to compare it to quantum algorithms such as the variational eigensolver \cite{cerezo2021variational,Ferguson2021,chan2023hybrid} and adiabatic simulation techniques \cite{Tameem2018,Elgart2021,Jansen2007}, as their runtimes are problem-dependent and there is generally no certainty in terms of the quality of the obtained solution.

\subsection{The hydrogen atom}\label{sec:chemistry}

In this section, we employ {\algoAE} to determine the ground state energy of the hydrogen molecule H\textsubscript{2} for varying atomic distances. The scope is to showcase two aspects of our approach. First, it always yields the desired result within the specified precision $\epsilon$. This is crucial for more advanced problems where the exact ground state is not available, and as such having estimates with a user-specified error is desirable. We remark that other quantum approaches such as VQE and adiabatic evolution do not generally have access to the error nor insurance on the quality of the output \cite{DePalma2023}.

The second aspect pointed out here is that the runtime of {\algoAE} is completely determined and known beforehand, as explained in \Cref{sec:alg}. This is again relevant in comparison to other quantum algorithms for which the required resources to yield an accurate result are not accessible beforehand. For instance, the runtime of VQE depends on the chosen ansatz state, resource (parametrized) circuit, measurement protocol and classical minimizer \cite{Paulson2021}. These are often challenging to be accurately characterized due to, e.g., barren plateaus \cite{cerezo2021variational,Uvarov2020}. On the other hand, adiabatic approaches are limited by the energy gap between the ground and first excited state of the time-varying Hamiltonian, which is generally unknown and possibly vanishing \cite{Elgart2021,Jansen2007}.

An important remark is that while it can be challenging to estimate the runtimes of VQE and adiabatic evolution, for specific problems they may converge
faster than {\algoAE}. We expect adiabatic evolution to be advantageous when the energy gap between the two smallest eigenvalues is large \cite{Elgart2021,Jansen2007}, and VQE when the desired ground state can be parametrized by a resource circuit with few parameters compared to the number of degrees of freedom of the system \cite{DePalma2023}. 

\begin{figure}[htp]
\centering
\begin{tikzpicture}[scale = 0.85]

\definecolor{darkgray176}{RGB}{176,176,176}

\begin{axis}[
tick align=outside,
tick pos=left,
x grid style={darkgray176},
xlabel={Bond length (\AA)},
xmin=0.195, xmax=2.65,
xtick style={color=black},
y grid style={darkgray176},
ylabel={Ground state energy (Ha)},
ymin=-1.16434910918163, ymax=-0.574129080998769,
ytick style={color=black}
]
\addplot [only marks]
table {%
0.3 -0.60095726409799
0.4 -0.910532640562286
0.5 -1.05703726445931
0.6 -1.11400684688828
0.7 -1.13752092608241
0.8 -1.13547674597398
0.9 -1.12187343787964
1 -1.10244074077586
1.1 -1.08045762420146
1.2 -1.05797911436733
1.3 -1.03639937534023
1.4 -1.01665825114288
1.5 -0.99931905974501
1.6 -0.98462523613751
1.7 -0.972565079108878
1.8 -0.962944082034135
1.9 -0.955457219843257
2 -0.949752800979517
2.1 -0.945481370199179
2.2 -0.942327030607747
2.3 -0.940022685658354
2.4 -0.938353298657687
2.5 -0.9371518593149297
};
\addplot [semithick]
table {%
0.3 -0.6018037107656857
0.31 -0.6458357415434062
0.32 -0.6862795104353605
0.33 -0.7234729786886129
0.34 -0.7577147158832309
0.35 -0.7892693924044143
0.36 -0.8183723740907125
0.37 -0.8452335886335343
0.38 -0.8700407973261428
0.39 -0.8929623781271656
0.4 -0.9141497046270843
0.41 -0.9337391888604558
0.42 -0.9518540428630697
0.43 -0.9686058035926723
0.44 -0.9840956576839244
0.45 -0.9984155960160175
0.46 -1.0116494228720427
0.47 -1.0238736402861794
0.48 -1.0351582247914952
0.49 -1.0455673110333195
0.5 -1.055159794470625
0.51 -1.0639898635485878
0.52 -1.072107470210788
0.53 -1.0795587463646923
0.54 -1.086386372869867
0.55 -1.0926299067451062
0.56 -1.0983260715554028
0.57 -1.1035090153172284
0.58 -1.1082105397311033
0.59 -1.1124603040961543
0.6 -1.1162860068695404
0.61 -1.1197135474938045
0.62 -1.1227671708179052
0.63 -1.125469596176865
0.64 -1.1278421329637447
0.65 -1.129904784322914
0.66 -1.131676340410248
0.67 -1.1331744625019833
0.68 -1.1344157590862745
0.69 -1.135415854939085
0.7 -1.1361894540659225
0.71 -1.136750397283203
0.72 -1.1371117151154644
0.73 -1.1372856765971084
0.74 -1.1372838344885021
0.75 -1.137117067345731
0.76 -1.1367956188204549
0.77 -1.1363291345101947
0.78 -1.1357266966302222
0.79 -1.1349968567347246
0.8 -1.134147666677096
0.81 -1.1331867079665368
0.82 -1.132121119649999
0.83 -1.1309576248253252
0.84 -1.1297025558704243
0.85 -1.1283618784581106
0.86 -1.126941214412047
0.87 -1.1254458634489712
0.88 -1.12388082384415
0.89 -1.1222508120508188
0.9 -1.1205602812999884
0.91 -1.118813439204268
0.92 -1.1170142643875662
0.93 -1.11516652216188
0.94 -1.113273779272859
0.95 -1.1113394177361497
0.96 -1.1093666477880977
0.97 -1.1073585199755278
0.98 -1.1053179364114252
0.99 -1.1032476612245001
1.0 -1.1011503302326195
1.01 -1.099028459871535
1.02 -1.0968844554116517
1.03 -1.0947206184967793
1.04 -1.0925391540401934
1.05 -1.0903421765127632
1.06 -1.0881317156598558
1.07 -1.085909721682727
1.08 -1.0836780699206483
1.09 -1.0814385650690592
1.1 -1.0791929449690747
1.11 -1.0769428840023643
1.12 -1.0746899961245715
1.13 -1.0724358375693206
1.14 -1.0701819092534717
1.15 -1.0679296589128593
1.16 -1.0656804829962414
1.17 -1.0634357283435236
1.18 -1.0611966936728123
1.19 -1.0589646308990266
1.2 -1.056740746305258
1.21 -1.054526201586371
1.22 -1.0523221147826762
1.23 -1.0501295611199792
1.24 -1.047949573770767
1.25 -1.0457831445498011
1.26 -1.0436312245560424
1.27 -1.041494724771481
1.28 -1.0393745166262496
1.29 -1.0372714325382197
1.3 -1.035186266434255
1.31 -1.0331197742593088
1.32 -1.0310726744786842
1.33 -1.029045648577734
1.34 -1.027039341563174
1.35 -1.0250543624687884
1.36 -1.0230912848680698
1.37 -1.0211506473960332
1.38 -1.019232954281587
1.39 -1.0173386758917071
1.4 -1.0154682492882452
1.41 -1.013622078797931
1.42 -1.0118005365958733
1.43 -1.010003963302642
1.44 -1.0082326685948688
1.45 -1.0064869318291225
1.46 -1.004767002678702
1.47 -1.0030731017829186
1.48 -1.00140542140833
1.49 -0.9997641261213248
1.5 -0.9981493534714099
1.51 -0.9965612146845115
1.52 -0.9949997953655102
1.53 -0.9934651562094206
1.54 -0.9919573337202128
1.55 -0.9904763409366258
1.56 -0.9890221681641085
1.57 -0.9875947837120684
1.58 -0.9861941346356116
1.59 -0.9848201474808991
1.6 -0.9834727290331732
1.61 -0.9821517670667529
1.62 -0.9808571310959652
1.63 -0.9795886731261829
1.64 -0.9783462284040639
1.65 -0.9771296161660871
1.66 -0.9759386403844545
1.67 -0.9747730905094947
1.68 -0.9736327422075991
1.69 -0.9725173580938391
1.7 -0.9714266884583405
1.71 -0.9703604719855032
1.72 -0.9693184364651972
1.73 -0.9683002994951055
1.74 -0.9673057691733091
1.75 -0.9663345447803094
1.76 -0.9653863174496992
1.77 -0.9644607708266899
1.78 -0.963557581713767
1.79 -0.962676420702742
1.8 -0.9618169527925812
1.81 -0.9609788379923183
1.82 -0.9601617319085155
1.83 -0.9593652863167003
1.84 -0.9585891497162963
1.85 -0.957832967868579
1.86 -0.9570963843173439
1.87 -0.9563790408918068
1.88 -0.9556805781915433
1.89 -0.9550006360532033
1.9 -0.9543388539987278
1.91 -0.9536948716651354
1.92 -0.9530683292155393
1.93 -0.9524588677315093
1.94 -0.9518661295867381
1.95 -0.9512897588020882
1.96 -0.9507294013821042
1.97 -0.9501847056331403
1.98 -0.9496553224633009
1.99 -0.9491409056643919
2.0 -0.9486411121761864
2.01 -0.948155602333259
2.02 -0.9476840400947564
2.03 -0.9472260932574239
2.04 -0.9467814336523228
2.05 -0.9463497373256127
2.06 -0.9459306847038376
2.07 -0.9455239607441582
2.08 -0.9451292550700428
2.09 -0.9447462620928349
2.1 -0.9443746811197424
2.11 -0.9440142164487292
2.12 -0.9436645774508124
2.13 -0.9433254786403016
2.14 -0.9429966397334844
2.15 -0.9426777856962743
2.16 -0.9423686467813528
2.17 -0.9420689585553064
2.18 -0.9417784619162844
2.19 -0.9414969031026674
2.2 -0.9412240336932627
2.21 -0.9409596105994985
2.22 -0.940703396050122
2.23 -0.9404551575688458
2.24 -0.940214667945422
2.25 -0.9399817052005904
2.26 -0.9397560525453218
2.27 -0.9395374983348027
2.28 -0.9393258360175442
2.29 -0.9391208640800293
2.3 -0.9389223859872748
2.31 -0.9387302101196692
2.32 -0.9385441497064106
2.33 -0.9383640227559833
2.34 -0.9381896519838554
2.35 -0.9380208647377856
2.36 -0.9378574929210381
2.37 -0.9376993729137434
2.38 -0.9375463454926982
2.39 -0.9373982557498383
2.4 -0.9372549530096294
2.41 -0.937116290745592
2.42 -0.9369821264961803
2.43 -0.9368523217801993
2.44 -0.9367267420119304
2.45 -0.9366052564161927
2.46 -0.9364877379434259
2.47 -0.9363740631850049
2.48 -0.936264112288885
2.49 -0.9361577688757258
2.50 -0.9360549199556065
};
\end{axis}
\begin{axis}[
xshift=.15\textwidth,
yshift=3.5cm,
width=0.35\textwidth,
height=0.2\textwidth,
log basis y={2},
tick align=outside,
tick pos=left,
x grid style={darkgray176},
xmin=0.19, xmax=2.61,
xtick style={color=black},
xtick={0.5,1.5,2.5},
y grid style={darkgray176},
ylabel={Error},
ymin = 0.00005, ymax = 0.02,
ymode=log,
ytick style={color=black},
ytick={0.01, 0.001, 0.0001},
yticklabels={
  \(\displaystyle {10^{-2}}\),
  \(\displaystyle {10^{-3}}\),
  \(\displaystyle {10^{-4}}\)
}
]
\addplot [only marks, mark options = {scale = 0.75}]
table {%
0.3 0.000846446667695844
0.4 0.00361706406479767
0.5 0.00187746998868521
0.6 0.00227915998125572
0.7 0.00133147201648343
0.8 0.00132907929688719
0.9 0.00131315657964826
1 0.00129041054324164
1.1 0.0012646792323856
1.2 0.00123836806207622
1.3 0.0012131089059777
1.4 0.00119000185463469
1.5 0.00116970627359936
1.6 0.0011525071043359
1.7 0.00113839065053711
1.8 0.00112712924155378
1.9 0.00111836584452984
2 0.00111168880333123
2.1 0.00110668907943701
2.2 0.00110299691448446
2.3 0.00110029967107883
2.4 0.00109834564805844
2.5 0.0010969393593232
};
\addplot [semithick, dashed]
table {%
0.3 0.01
0.4 0.01
0.5 0.01
0.6 0.01
0.7 0.01
0.8 0.01
0.9 0.01
1 0.01
1.1 0.01
1.2 0.01
1.3 0.01
1.4 0.01
1.5 0.01
1.6 0.01
1.7 0.01
1.8 0.01
1.9 0.01
2 0.01
2.1 0.01
2.2 0.01
2.3 0.01
2.4 0.01
2.5 0.01
};
\end{axis}

\end{tikzpicture}
\caption{\textit{Main plot:} Ground state energy of the hydrogen molecule H\textsubscript{2} for different bond lengths, obtained by employing {\algoAE} with precision $\epsilon = 10^{-2}$ (dots), compared to exact diagonalization (solid curve). \textit{Inset:} The corresponding errors (dots) and the chosen precision $\epsilon = 10^{-2}$ (dashed line).}
\label{fig:h2}
\end{figure}
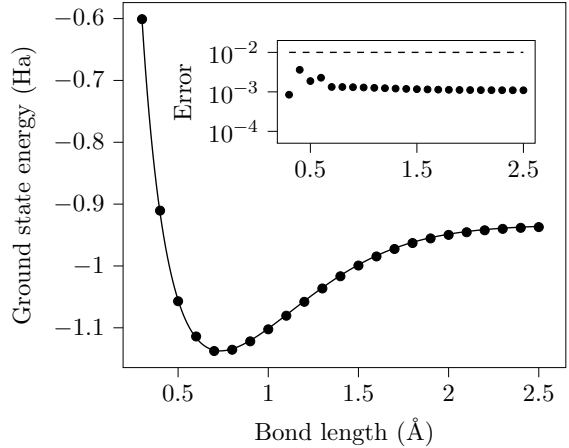
In \Cref{fig:h2} we present our results from determining the ground state energy of the H\textsubscript{2} molecule with different bond lengths. We apply {\algoAE} to the Hamiltonians in Ref.~\cite{mcclean2019openfermion} and calculate the lowest eigenvalues, depicted as dots in the figure. The line comes from exact diagonalization, and is used as a reference. As it is possible to see in the inset, the absolute error (dots) affecting each outcome from {\algoAE} lies well below the chosen threshold of $\epsilon$ (horizontal dashed line). Knowing $\epsilon$, i.e., the \textit{maximum} possible error affecting the estimate of the ground state energy, is of foremost importance for most applications. This is particularly true when considering quantum approaches that may not ensure the correctness of the result or may require a massive experimental budget due to the statistical noise that is intrinsic to quantum mechanics \cite{Shlosberg2023adaptiveestimation}.

\subsection{$1D$ Elasticity}\label{sec:Maxwell_waves}
In this section, we demonstrate how {\QAA} is used to determine a state $\rho$ lying in the low-energy subsector of the system Hilbert space (see \Cref{prob:eig} and \Cref{sec:alg} for more details). Specifically, we consider the elasticity problem as a set of PDEs that can be solved using {\algoAE} alongside {\QAA}. As schematically represented in \Cref{fig:cavity} (top), the goal is to determine the fundamental vibrational mode (bottom left panel of the figure) of a composite string made of two different materials. While the $1D$ case investigated here constitutes a proof of principle, studying the behaviour of (composite) materials is at the basis of countless applications \cite{zienkiewicz2005finite,llorca2011multiscale}, including structural integrity and material simulations. For large systems with complex geometries, the equations are extremely challenging to be solved classically \cite{sukumar2001modeling,agathos2019unified}, making the development of novel approaches such as the one presented here of paramount importance.

As shown in \Cref{fig:cavity} (top), we assume a $1D$ composite string attached to fixed walls. Once expressed in terms of the displacement $\Phi$, the system can be described via the operator $\mathcal{L}$ and its eigenvalues $\lambda$ as
\begin{subequations}\label{eq:PDE}
\begin{align}
    &\mathcal{L}\Phi = \lambda \Phi
    \text{, where }
    \label{eq:PDE_eigenvalue_problem}\\
    &\mathcal{L}\Phi(z) = \frac{1}{\epsilon_r\left( z \right)} \frac{\partial^2}{\partial z^2} \Phi\left( z \right)
    \label{eq:PDE_maxwell}.
\end{align}
\end{subequations}
Here, $z$ is the spatial coordinate, and $\epsilon_r\left( z \right) = \tilde{\epsilon}_r$ (constant) for $z\in[z_0-\frac{d}{2},z_0+\frac{d}{2}]$, otherwise it is equal to one. $\tilde{\epsilon}_r$ is the ratio between the coefficients characterizing the middle and the two outer sections of the string, represented in green and gray, respectively, in \Cref{fig:cavity} (top).

To find the lowest energy eigenvector of the considered system, we employ {\algoAE} to determine an estimator $\tilde{\lambda_\mathcal{L}}$ of $\lambda_\mathcal{L}$, where we formally define
\begin{equation}\label{eq:lambda_def}
    \lambda_\mathcal{L} = \min\{\lambda: \exists \Phi \text{ with } \mathcal{L}\Phi = \lambda\Phi\}.
\end{equation}
We then use {\QAA} to find the associated eigenfunction $\Phi$, which corresponds to the mechanical mode with the smallest frequency between the walls. 

In \Cref{sec:Maxwell_waves_theory} we employ the finite element method \cite{zienkiewicz2005finite} to cast this problem in a form similar to \Cref{prob:eig}. Then, we discuss our numerical results and give a comparison with other approaches in \Cref{sec:Maxwell_waves_results}.

\subsubsection{Semi-analytical solution and discretization of the problem}\label{sec:Maxwell_waves_theory}
To validate the output of {\algoAE}, we first provide a semi-analytical solution for the eigenvalue operator in \Cref{eq:PDE_maxwell}. We start by considering that within the three sectors dividing the string (see \Cref{fig:cavity} (top)), $\epsilon_r\left( z \right)$ is constant and as such the solution is a wave. By imposing the boundary conditions $\Phi(\pm 1) = 0$, continuity at all interfaces and conservation of the stress 
\begin{equation}
\label{eq:stress_cons}
    \lim_{z \downarrow z_0\pm\frac{d}{2}}\frac{\partial{\Phi(z)}}{\partial z}
    = \lim_{z \uparrow z_0\pm\frac{d}{2}}\epsilon_r^{\pm 1}
    \frac{\partial{\Phi(z)}}{\partial z}
\end{equation}
all solutions can be found, provided a transcendental equation for the eigenvalues is numerically solved (hence this method is semi-analytical). 
In the previous \Cref{eq:stress_cons}, with $x \downarrow y$ ($x \uparrow y$) we indicated that the argument $x$ is evaluated in $y$, approaching from the left ($\uparrow$) or right ($\downarrow$). 

To determine $\lambda_\mathcal{L}$ in \Cref{eq:lambda_def} and a quantum state $\rho$ encoding the corresponding solution $\Phi$ of \Cref{eq:PDE_maxwell}, the first step is the discretization of $\mathcal{L}$. There are several approaches to doing so, the most popular of which include finite difference and finite element methods \cite{grasselli2008numerical}. While both are in principle valid, we employ the latter because the boundary conditions in \Cref{eq:stress_cons} are automatically satisfied. Furthermore, in the case where the material geometry is more complex, finite element methods may be preferable because of their flexibility \cite{zienkiewicz2005finite}. We remark that {\algoAE} is applicable in both cases.

Our discretization process begins by selecting a number $N+2$ of grid points $\{ p_i \}_{i=0}^{N+1}$ and the corresponding spacing $h = 2/(N+1)$ to be used. These points are equally spaced between $-1$ and $1$, and the $N$ ones obtained by ignoring the extrema are collected in the set $P=\{p_n = -1 + n h : n =1,\ldots,N\} = \{-1+h, -1+2h, \ldots, 1-2h, 1-h\}$. Since at $p_0 = -1$ and $p_{N+1} = 1$ the boundary conditions fully characterize the solution $\Phi$, $P$ contains all relevant points that must be considered.

Following Ref.~\cite{zienkiewicz2005finite}, we choose a basis of functions to approximate the exact solution $\Phi$. We opt for the tent functions, i.e., piecewise connected segments that are zero everywhere except for a triangular hat of unit height centered at a given grid point and with a width of $2h$. 

With the basis of the tent functions, it is possible to approximate the operator $\mathcal{L}$ in \Cref{eq:PDE_maxwell} with an 
$N\times N$ hermitian matrix $\mathcal{D}$. Its elements $\mathcal{D}_{ij}$ are
\begin{subequations}\label{eq:D_matrix}
\begin{align}
    \mathcal{D}_{ii} & =
    -\frac{h^{-1}}{\epsilon_{r}\left(p_{i-1}+\frac{h}{2}\right)}
    -
    \frac{h^{-1}}{\epsilon_{r}\left(p_{i}+\frac{h}{2}\right)}
    \label{eq:D_matrix_diag},
    \\
    \mathcal{D}_{ij} & = \frac{h^{-1}}{\epsilon_{r}\left(p_{i}+\frac{h}{2}\right)}
    \text{ if }
    |i-j|=1
    \label{eq:D_matrix_offdiag},
\end{align}
\end{subequations}
and $0$ for all other $i,j=1,\ldots,N$. We fix the parameters $z_0$ and $d$ defined below \Cref{eq:PDE_maxwell} such that the coordinates at which the string material changes are included in $P$. While this is not a requirement for finite element methods \cite{llorca2011multiscale}, it simplifies the following derivations.

The matrix defined in \Cref{eq:D_matrix} corresponds to the operator $\mathcal{L}$ in \Cref{eq:PDE_maxwell} in the following sense. For a function $\Phi:[-1,1]\rightarrow \mathbb{R}$, let 
\begin{equation}\label{eq:discretized_func}
    \Vec{\Phi} = (\Phi_1, \ldots, \Phi_N)^T
\end{equation}
be the vector of values of $\Phi$ at the gridpoints in $P$. These can be viewed as the coefficients of the tent functions employed to approximate $\Phi$. Then, provided $h$ is sufficiently small, $\mathcal{D}\vec{\Phi}$ is equivalent to the weak form of the eigenvalue equation in \Cref{eq:PDE_maxwell}. Therefore, $\mathcal{L}\Phi = \lambda_{\mathcal{L}} \Phi$ if and only if $\mathcal{D} \Vec{\Phi} = \lambda_{\mathcal{D}} \Vec{\Phi}$, where $\lambda_{\mathcal{D}}$ and $\lambda_{\mathcal{L}}$ are the corresponding eigenvalues of the operator in the weak and strong form, respectively \cite{zienkiewicz2005finite}. 

In the following section, we first employ {\algoAE} to estimate $\lambda_{\mathcal{D}}$. Then, we apply {\QAA} to find the state $\rho$ that approximates $\rho_{\rm e}$, i.e., the density matrix obtained from the desired eigenvector of the matrix $\mathcal{D}$. As a final remark, we point out that while finite element methods allow for arbitrary boundary conditions \cite{zienkiewicz2005finite}, the form of $\mathcal{D}$ in \Cref{eq:D_matrix} already incorporates the chosen ones.

\subsubsection{Numerical results and comparison with other approaches}\label{sec:Maxwell_waves_results}
In this section, we first present numerical results of {\algoAE} applied to the composite string depicted in \Cref{fig:cavity}, and later on discuss about its performance compared to other approaches. 

In \Cref{fig:cavity} (bottom left) we show the expected $\Phi$ determined with our semi-analytical approach (full line) with the associated eigenvector of $\mathcal{D}$ with $N=2^4$ (dots) and the sine function corresponding to the case $\tilde{\epsilon}_r=1$ (dashed line). The area where the string material is changed is highlighted by the green dotted line. We also report the value $\lambda_{\mathcal{D}} \approx -33.09$ found by exact diagonalization of the matrix $\mathcal{D}$. 

To determine $\rho$, as a first step we must run {\algoAE} to find an estimate $\tilde{\lambda}_{\mathcal{D}} \approx -33.06$. Consistently with the chemistry example in \Cref{sec:chemistry}, the difference $|\tilde{\lambda}_{\mathcal{D}} - \lambda_{\mathcal{D}}|$ is well below the error associated to the $5$ clock qubits we employed. 

Once $\tilde{\lambda}_{\mathcal{D}}$ is known, the last step is to run {\QAA} to determine $\rho$. The output from our algorithm is presented in the two plots in the bottom right of \Cref{fig:cavity}. Specifically, $\rho$ is in the upper one while the difference between $\rho$ and the exact result $\rho_{\rm e}$ (found via exact diagonalization) in the lower. The latter one in particular shows that $\rho$ is, to a high approximation, the desired one, as it is further confirmed by the fidelity of $\approx 99.96\%$. The fidelity can be further enhanced by setting lower values for $\epsilon$.

Above, we demonstrated that {\algoAE} in addition with {\QAA} can determine both the smallest eigenvalue $\lambda_\mathcal{D}$ and, to a high fidelity, the associated state $\Phi$. Below, we provide a qualitative discussion about its performance in comparison to other quantum and classical algorithms. We remark that the PDEs considered here should be taken as a proof of principle for the validity our approach, as there are classical algorithms that outperform {\algoAE}. In fact, due to limitations in the (quantum) computational power, we were not able to test our protocol in scenarios where large values of $N$ and small values of $||\mathcal{D}||$ ensure better performances than any other classical approach.

Other quantum approaches that are specialized in the solution of PDEs are in, e.g., Refs.~\cite{berry2014pde, Childs2021highprecision, papageorgiou2007sturm, papageorgiou2005classical,Kerenidis2020,vanApeldoorn2020quantumsdpsolvers,mande2023tight}. In particular, for certain families of PDEs, the algorithm from Childs et al., \cite{Childs2021highprecision}, presents very good scaling in terms of time complexity. (See also~\cite[Table~1]{Childs2021highprecision}.)
However, we remark that instead of determining the eigenvectors and eigenvalues of an operator, it finds a solution of a PDE. Therefore, it is not immediately applicable to the examples considered in this section. 

A technique of Jaksch and Papageorgiou \cite{jaksch2003eigenvector} is also applicable to the eigenvalue problem, and may be favourable as it only relies on QPE but neither QAE or QAA. The idea is to first find a classical approximation of the desired solution on a coarse grid, that is then used as input to QPE. As long as this coarse grid solution is sufficiently close to the true solution, QPE will give the corresponding eigenvalue with high probability. This approach obviates the need for iterations, as QPE only needs to be executed once (or a few times), and therefore its dependence on $N$ can be polylogarithmic. However, it is also limited by the requirement of rescaling the matrix $\mathcal{D}$ by its norm $||\mathcal{D}||$. Furthermore, the classical solution obtained on a coarse grid that is required as input to these algorithms is only ensured to yield a good starting point for specific kinds of PDEs, such as ones in Sturm-Liouville form \cite{papageorgiou2007sturm, papageorgiou2005classical}. 
For more general problems, it is unclear whether a classically pre-computed coarse grid solution is is sufficiently close to the true solution for this method to work.

For what concerns classical approaches, there exist several that determine the smallest eigenvalue (and the corresponding eigenvector) of a sparse matrix $\mathcal{D}$, including notably the Lanczos algorithm \cite{lanczos1950iteration}, which can run in linear time. While linear-time classical techniques seemingly scale worse than {\algoAE} (whose runtime is in $O(||\mathcal{D}|| \sqrt{N}/\epsilon)$), they are generally polylogarithmic in $1/\epsilon$. Depending on the problem at hand, they may therefore be advantageous if the magnitude of $||\mathcal{D}||/\epsilon$ is sufficiently large in relation to $N$. 

The situation changes for all problem instances where $||\mathcal{D}||/\epsilon$ increases slowly in $N$. {\algoAE} may thus be advantageous in scenarios where a large number of grid points is required to accurately model the system up to fixed error $\epsilon$. In this case, the speedup that our algorithm gives with respect to $N$ may overcome its slower runtime with respect to $\epsilon$. Such scenarios may occur when non-constant coefficients or boundary conditions are particularly complicated, when the desired solution is poorly behaved (in terms of continuity or differentiability), when there is a very large number of equations to be solved, and/or the variable space has many dimensions. 

We remark that in practice, $N$, $\epsilon$, and $\mathcal{D}$ are not independent of one another. When discretizing a PDE, the quantity $N$ and discretization $\mathcal{D}$ are often chosen so as to obtain a specific error $\epsilon$. Or conversely, if the number of grid points and type of discretization $\mathcal{D}$ are selected a priori, then the resulting error $\epsilon$ is completely determined. We believe that finding specific scenarios in which the scaling of our algorithm provides a provable speedup is an exciting open problem.

\begin{figure}
    \centering\includegraphics[scale=0.9]{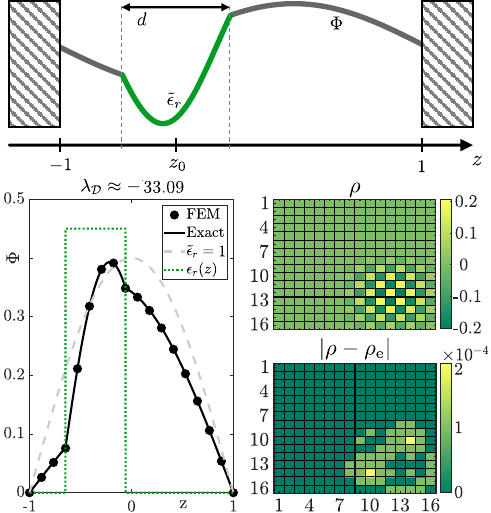}
    \caption{
    \textit{Top:} Schematic representation of the considered system.
    \textit{Bottom left:} Semi-analytical (full line) and finite element method (FEM, dots) solutions, compared with the sinusoidal shape obtained for a homogeneous string ($\tilde{\epsilon}_r = 1$, dashed line). The region where the string material is changed to $\tilde{\epsilon}_r = 5$ is indicated by the green dotted line, and we report the value $\lambda_{\mathcal{D}}$.    
    \textit{Bottom right:} Density matrix $\rho$ found with our approach (above) and its absolute difference with respect to the exact one $\rho_{\rm e}$ found with exact diagonalization of $\mathcal{D}$ (below). 
    We set $N=2^4$, $z_0 = -6/17$, $d=10/17$ and the precision $\epsilon$ within {\algoAE} is set by $5$ qubits.
    }
    \label{fig:cavity}
\end{figure}

\section{Conclusions and future work}
We presented a quantum algorithm for determining the lowest eigenvalue of a Hermitian matrix. Our algorithm makes use of Quantum Phase Estimation and Quantum Amplitude Estimation to achieve a quadratic speedup with respect to the best classical algorithm in terms of matrix dimensionality, requiring just $\tilde{\mathcal{O}}(\sqrt{N}/\epsilon)$ black-box queries to an oracle encoding the matrix, where $N$ is the matrix dimension and $\epsilon$ is the desired precision. In contrast, the best classical algorithm for the same task requires $\Omega(N)$polylog($1/\epsilon$) queries. We present simulations of our algorithm for two problem instances: the Hydrodgen atom and one-dimensional elasticity.

Because of the generality of this algorithm, we believe that there may be many interesting problems to which it could be applied. It would also be interesting to investigate the exact scaling of this algorithm for the two applications from \Cref{sec:applications} (or other ones), in order to delineate the precise scenarios in which a quantum advantage is obtained.

\textit{Authors' note: After developing \Cref{sec:alg}, we became aware of the works in Refs.~\cite{Kerenidis2020,vanApeldoorn2020quantumsdpsolvers}, who outlined the same steps as {\algoAE} as subroutines for quantum gradient descent and for aproximately solving semidefinite programs, respectively. Compared to these works, we focus on determining the smallest eigenvalue and the associated ground state, eliminate residual populations via the median trick and provide various numerical simulations to showcase {\algoAE}.}

\section*{Acknowledgments}
We thank David Gosset for fruitful discussions and comments, and Kostantinos Agathos for his fundamental input on finite element methods and their applications. AK thanks Anne Broadbent for comments regarding the query complexity model. LD acknowledges the EPSRC quantum career development grant EP/W028301/1.

\printbibliography

\appendix
\section{Extension to general sparse Hermitian matrices}\label{app:extension_general}

\begin{proof}[Proof of \Cref{cor:gen1}]
    We first add $\tilde{||H||} I$ to $H$ to ensure all eigenvalues lie in $[0,2\tilde{||H||}]$, and then multiply by $1/2\tilde{||H||}$ to get eigenvalues in $[0,1]$. Let us call the resulting matrix $H'$. Note that $||H'||_\text{max} \leq 1$.

    The desired eigenvalue of $H$ has now been mapped to the lowest eigenvalue of $H'$. We will run this modified version of {\algoAE} using precision $\epsilon/2\tilde{||H||}$, so that when we undo the shift-and-rescale procedure, we find a final error of $\epsilon$ in our estimate of $\min\{\lambda:\lambda\in\text{eigs}(H)\}$. This gives us a total runtime of $\widetilde{O}(\sqrt{N}\tilde{||H||}/\epsilon)$. Next, we must find the corresponding eigenvector. The eigenvectors of $H'$ are the same as the ones of $H$. We may therefore run {\QAA} similarly to what was done in the proof of \Cref{thm:main} to determine $\rho$.
\end{proof}

An interesting extension of {\algoAE}, is that it can be easily modified to find other eigenvalues rather than the smallest. For example, if $H$ has both positive and negative eigenvalues, then we can find the smallest eigenvalue in absolute value. This is done by remarking that the shift-and-rescale procedure described above maps that eigenvalue to the eigenvalue of $H'$ nearest to $1/2$. We can then modify {\algoAE} such that we search for $\min\{\lambda\in\text{eigs}(H'): \lambda \geq 1/2\}$. Doing so \textit{only} requires us to modify the definition of $\chi_y$ so that $\chi_y(x) = 0$ for all bit strings $x$ whose most significant bit is $0$. In other words, we simply ignore any eigenvalues of $H'$ which are less than $1/2$. This will result in $y_{m+1}$ being an approximation of the smallest eigenvalue of $H'$ that is at least $1/2$. Similarly, we can search for the largest eigenvalue of $H'$ that is less than or equal to $1/2$, and then we can compare the two values. We remark that this extension can be generalized to any value that can be represented by the $m$ qubits employed (i.e., the points $k 2^{-m}$ with $k=0,\ldots, 2^m-1$ lying in the $[0,1]$ interval). 

\section{Technical details}\label{app:details}

In this appendix we provide more technical details regarding \algoAE. Let us begin with explaining our choice of $\delta = \frac{1}{2N+2}$.

Due to the inexactness of QPE, the probability of observing $X<\lambda_0 - \epsilon$ when observing the median register is nonzero. How can we distinguish values of $X$ which are small due to $\lambda_0$ being small from values of $X$ which are small due to QPE imprecision? We do so by ensuring that the contributions of the latter are much smaller than the contributions of the former.
For each eigenvalue $\lambda_j$, let $\mathcal{X}_\text{good}^j = \{x:|0.x_1x_2\ldots x_t - \lambda_j|<\epsilon/2\}\subset \{0,1\}^t$, and let $\mathcal{X}_\text{bad}^j = \{0,1\}^t\setminus \mathcal{X}_\text{good}^j$. Formally, we require that $\Pr[X < \lambda_0 - \epsilon/2] \leq \frac{1}{2}\Pr[|X -\lambda_0| < \epsilon/2]$. We achieve this by enforcing the stronger condition that
\begin{equation}
    \Pr\left[X \not\in \bigcup_j \mathcal{X}_\text{good}^j\right] < \frac{1}{2}\Pr\left[X \in \mathcal{X}_\text{good}^0\right],
\end{equation}
which is itself enforced by requiring
\begin{equation}\label{eq:separation}
     \sum_{j=0}^{N-1} \sum_{X\in \mathcal{X}_\text{bad}^j}|\braket{X|\tilde{\lambda_j}}|^2 < \frac{1}{2}\sum_{X\in \mathcal{X}_\text{good}^0} |\braket{X|\tilde{\lambda_0}}|^2.
\end{equation}
The left hand side of this inequality is at most $\delta N$, and the right hand side is at least $\frac{1-\delta}{2}$, hence taking $\delta = \frac{1}{2N+2}$ satisfies the inequality.

Next, let us prove correctness of {\algoAE} by taking $i=m$ in the following lemma (cf. \Cref{fig:convergence}).

\begin{lemma}\label{lem:minfindinginduction}
    For all $i=0,\ldots,m$, we have $|y_i - \lambda_0| \leq 1/2^{i+1} + \epsilon/2$, with $y_i$ as defined in \algoAE.
\end{lemma}

\begin{proof}[Proof]
    In order to prove the theorem, we need a technical lemma, which quantifies the possible values of $\widetilde{p}$ when {\QAE} is used on $\mathcal{A}$ and a function of the form $\chi_y$.
    \begin{lemma}\label{lem:thresholds}
    Let $k\geq 537, M = \sqrt{\frac{kN}{1-\delta}}$, and
    \begin{equation}\label{eq:threshold}
        q = \frac{1-\delta}{N}\left(\frac{1}{2} + \frac{\sqrt{2}\pi}{\sqrt{k}} + \frac{\pi^2}{k}\right).
    \end{equation}
    For any $\mathcal{B}, \chi, p_\mathrm{g}, \widetilde{p_\mathrm{g}}$ as in \Cref{thm:AE}, if $p_\mathrm{g} < \frac{1-\delta}{2N}$ then $\widetilde{p_\mathrm{g}} < q$, and if $p_\mathrm{g} > \frac{1-\delta}{N}$ then $\widetilde{p_\mathrm{g}} > q$.
\end{lemma}

Informally, this lemma is telling us that whenever $y \not\in [\lambda_0-\epsilon, \lambda_0 + \epsilon]$, {\QAE} with parameters $\mathcal{A}, \chi_y, M$ can reliably determine whether $y<\lambda_0$ or $y >\lambda_0$. In \Cref{fig:QAE_error}, we illustrate the fact that even though {\QAE} will make some amount of error, because of our stipulation that $\delta = \frac{1}{2N+2}$, 
its output will always lie on the correct side of the threshold value $q$, provided $y \notin [\lambda_0 - \epsilon, \lambda_0 + \epsilon]$. In the case where $y \in [\lambda_0-\epsilon, \lambda_0 + \epsilon]$, we have no guarantee as to whether $\tilde{p_\mathrm{g}}<q$ or $\tilde{p_\mathrm{g}}>q$, and so there is no guarantee as to whether $y_i$ will be obtained by adding $1/2^{i+1}$ from $y_{i-1}$ or subtracting $1/2^{i+1}$. However we will see shortly that in this case $y_{i-1}$ is close enough to $\lambda_0$ that it doesn't matter. 

\begin{figure}[htp]
    \centering
    \hspace*{-0.45cm}
    \begin{subfigure}{\linewidth}
    \centering
    \begin{tikzpicture}[scale=8, every node/.style={align=center}] 
    \draw[->, thick] (0,0) -- (1,0);
  
    \draw (0.0,0.025) -- + (0,-0.05) node[below] {$0$}; 
    \draw (0.3,0.025) -- + (0,-0.05) node[below] {$\frac{1-\delta}{2N}$}; 

    \draw (0.6,0.025) -- + (0,-0.05) node[below] {$\frac{1-\delta}{N}$}; 
    
    \draw (0.2,0.025) -- + (0,-0.05) node[below] {$p_1$}; 
    \draw (0.45, 0.025)  -- + (0, -0.05) node[below] {$q$};

    \draw[thick, decorate,decoration={brace,amplitude=2pt}]
(0.05,0.05) -- (0.35,0.05) node[above, xshift=-1.2cm, yshift=0.1cm] {Possible values of $\Tilde{p_1}$};
\end{tikzpicture}
    \end{subfigure}
    \\~\\~\\
    \hspace*{-0.45cm}
    \begin{subfigure}{\linewidth}
    \centering
    \begin{tikzpicture}[scale=8, every node/.style={align=center}] 
    \draw[->, thick] (0,0) -- (1,0);
  
    \draw (0.0,0.025) -- + (0,-0.05) node[below] {$0$}; 
    \draw (0.3,0.025) -- + (0,-0.05) node[below] {$\frac{1-\delta}{2N}$}; 

    \draw (0.6,0.025) -- + (0,-0.05) node[below] {$\frac{1-\delta}{N}$}; 
    
    \draw (0.7,0.025) -- + (0,-0.05) node[below] {$p_1$}; 
    \draw (0.45, 0.025)  -- + (0, -0.05) node[below] {$q$};

    \draw[thick, decorate,decoration={brace,amplitude=2pt}]
(0.5,0.05) -- (0.9,0.05) node[above, xshift=-1.6cm, yshift=0.1cm] {Possible values of $\Tilde{p_1}$};
\end{tikzpicture}
    \end{subfigure}
    \caption{Illustration of the fact that $p_1 < \frac{1-\delta}{2N}\Rightarrow \Tilde{p_1}<q$ and $p_1 > \frac{1-\delta}{N}\Rightarrow \Tilde{p_1} > q$ when QAE is run with the parameters specified in the main text. In each figure, we show a possible value of $p_1$ together with the range of values of $\Tilde{p_1}$ that may be output by QAE.}
    \label{fig:QAE_error}
\end{figure}
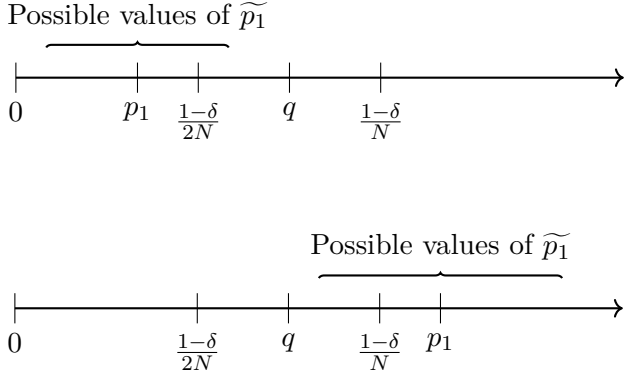

\begin{proof}[Proof of \Cref{lem:thresholds}]
    If $p_\mathrm{g} \leq \frac{1-\delta}{2N}$ we find (from \Cref{thm:AE})

    \begin{align}
        \widetilde{p_\mathrm{g}} &\leq p_\mathrm{g} + 2\pi\frac{\sqrt{p_\mathrm{g}(1-p_\mathrm{g})}}{M} + \frac{\pi^2}{M}
        \\
        & \leq \frac{1-\delta}{2N} + 2\pi\sqrt{\frac{1-\delta}{2N}\frac{1-\delta}{kN}} + \pi^2 \frac{1-\delta}{kN}
        \\
        & = \frac{1-\delta}{N}\left(\frac{1}{2} + \frac{\sqrt{2}\pi}{\sqrt{k}} + \frac{\pi^2}{k}\right),
    \end{align}
which proves the first half of the lemma. Now suppose $p_\mathrm{g} \geq \frac{1-\delta}{N}$. Let us write $p_\mathrm{g} = \alpha\frac{1-\delta}{N}$ for some $\alpha\geq 1$. We then find
    \begin{align}
        \widetilde{p_\mathrm{g}} &\geq \alpha\frac{1-\delta}{N} - 2\pi \sqrt{\alpha\frac{1-\delta}{N}\frac{1-\delta}{kN}} - \pi^2 \frac{1-\delta}{kN}\\
        &= \frac{1-\delta}{N}\left(\alpha - \frac{2\pi \sqrt{\alpha}}{\sqrt{k}} - \frac{\pi^2}{k}\right)\\
        &\geq \frac{1-\delta}{N}\left(1 - \frac{2\pi}{\sqrt{k}} - \frac{\pi^2}{k}\right)\\
        &\geq \frac{1-\delta}{N}\left(\frac{1}{2} + \frac{\sqrt{2}\pi}{\sqrt{k}} + \frac{\pi^2}{k}\right),
    \end{align}
    where the last inequality holds because $k\geq 537$.
\end{proof}

We will prove \Cref{lem:minfindinginduction} by induction on $i$. The base case is trivial: We always have $\lambda_0 \in [0, 1]$, hence $|y_0 - \lambda_0| \leq 1/2 < 1/2^1 + \epsilon/2$. Now we consider the inductive step. Suppose that $i \geq 1$ and that $|y_{i-1} - \lambda_0| \leq 1/2^i + \epsilon/2$, and let us now consider $y_i$. We will consider several cases separately.

    \begin{enumerate}
    \item First, suppose
    $y_{i-1} \geq \lambda_0 + \epsilon/2 + 1/2^{i+1}$. Then by \Cref{lem:thresholds}, we will set $y_i = y_{i-1} - 1/2^{i+1}$, and by the induction hypothesis,
    \begin{align}
       |y_i - \lambda_0| &= y_i - \lambda_0 \\
       &= y_{i-1} - 1/2^{i+1} - \lambda_0 \\ 
       &\leq \epsilon/2 + 1/2^i - 1/2^{i+1}\\
       &= \epsilon/2 + 1/2^{i+1},
    \end{align}
    as desired. The case $y_{i-1} \leq \lambda_0 - \epsilon/2 - 1/2^{i+1}$ is similar.
    \item Suppose next that $\lambda_0+\epsilon/2 < y_{i-1}< \lambda_0 + \epsilon/2 + 1/2^{i+1}$. Then by \Cref{lem:thresholds} we will again set $y_i = y_{i-1} - 1/2^{i+1}$. If $i=m$ then we will have $y_i \in [\lambda_0, \lambda_0-\epsilon/2]$ and we are done. Otherwise, we will have $y_i < \lambda_0$ and particularly, 
    \begin{align}
        |y_i - \lambda_0| &= \lambda_0-y_i\\
        &< y_{i-1} - y_i\\
        &< 1/2^{i+1},
    \end{align}
    as desired. The case $\lambda_0 - \epsilon/2 - 1/2^{i+1} < y_{i-1} < \lambda_0 - \epsilon/2$ is similar.
    \item The final case is that in which $|\lambda_0 - y_{i-1}| \leq \epsilon/2$. In this case, \Cref{lem:thresholds} gives us no guarantee as to whether $y_i$ will be  $y_{i-1} - 1/2^{i+1}$ or $y_{i-1} + 1/2^{i+1}$. But since $y_{i-1}$ is so close to $\lambda_0$ to begin with it turns out not to matter:
    \begin{align}
        |y_i - \lambda_0| &\leq |y_i - y_{i-1}| + |y_{i-1} - \lambda_0|\\
        & \leq 1/2^{i+1} + \epsilon/2.
    \end{align}
    \end{enumerate} 
\end{proof}

We now provide further details of our algorithm for preparing $\rho$.

As stated in \Cref{sec:alg}, our goal will be to find an estimate $\theta_0$ of $\lambda_0$, and then to project $\ket{\Psi}$ onto the space spanned by
\begin{equation}
    \{\ket{0^{ct}}\ket{x}\ket{\psi_j}\ket{\bar{\psi_j}} : x\approx \theta_0\}.
\end{equation}
We have not yet stipulated how close $\theta_0$ must be to $\lambda_0$, nor how close we will require $x$ and $\theta_0$ to be. If we require $x$ to be too close to $\theta_0$, then this projection could be a vector of very short length, which would cause {\QAA} to require a very large number of iterations. On the other hand, if we don't require that $x$ is very close to $\theta_0$, then states of the form $\ket{0^{ct}}\ket{x}\ket{\psi_j}\ket{\bar{\psi_j}}$ with $\lambda_j \not\approx \lambda_0$ will be amplified, which is undesirable.

With this in mind, we will first employ {\algoAE} to find $\theta_0$ such that $|\theta_0 - \lambda_0| < \epsilon/4$. Next, we use {\QAA}, setting $\mathcal{B} = \mathcal{A}$ in \Cref{thm:AA}, and using $\epsilon/4$ for the precision of the QPE component of $\mathcal{A}$. When running {\QAA}, we set $\chi$ to be the indicator function $\chi(x) = 1$ if and only if $|x - \theta_0|<\epsilon/2$.
 
Ignoring the $c$ uncomputed clock registers, let us write the resulting state as
\begin{equation}\label{eq:PhiforRho}
    \ket{\Phi} = \frac{\sum_{j=0}^{N-1} a_j\ket{\mu_j}\ket{\psi_j}\ket{\bar{\psi_j}}}{\sqrt{\sum_{j=0}^{N-1} |a_j|^2}},
\end{equation}
where $\frac{1-\delta}{N}\leq\sum_j |a_j|^2 \leq 1$, and $\ket{\mu_j}$ is a superposition of strings $x$ with $|x - \theta_0|<\epsilon/2$.

We now claim that tracing out the first and third register of $\ket{\Phi}$ results in a state $\rho$ satisfying $\Tr(\Pi \rho) > 2/3$, with $\Pi$ as in \Cref{prob:eig}. We refer the reader to \Cref{fig:QAA_numberline} as a supplement to the following argument. First, we deduce that $|a_0|^2\geq \frac{1-\delta}{N}$. This follows from the facts that (i) $[\lambda_0 - \epsilon/4, \lambda_0 + \epsilon/4] \subset [\theta_0-\epsilon/2, \theta_0+\epsilon/2]$; and (ii) we have set the precision of the QPE component of $\mathcal{A}$ to $\epsilon/4$. Second, for all $j$ such that $|\lambda_j - \lambda_0| > \epsilon$, we have $[\lambda_j -\epsilon/4, \lambda_j + \epsilon/4]\cap [\theta_0 - \epsilon/2, \theta_0 + \epsilon/2] = \emptyset$, which (recalling that QPE was run with precision $\epsilon/4$) implies $|a_j|^2 \leq \delta$. Due to our choice of $\delta = \frac{1}{2N+2}$ (see \Cref{eq:separation}), the ratio of $|a_0|^2$ to the sum of all such $|a_j|^2$ is at least $2$. If there are no $j\neq 0$ such that $|\lambda_j - \lambda_0| < \epsilon$, then we must therefore have $\Tr(\Pi \rho) > 2/3$. If there do exist any such $j$, then $\Tr(\Pi \rho)$ can only increase.

\begin{figure}
\centering
\begin{tikzpicture}[scale=8]
\draw[<->, dotted, thick] (0.05,0) -- (1,0);
\draw (0.125,0) -- + (0, -0.02) node[below, yshift = -0.1cm, label={[scale=0.8]center:$\lambda_0 - \frac{\epsilon}{4}$}] {};
\draw (0.25,0) -- + (0, -0.02) node[below, yshift = -0.1cm, label={[scale=0.8]center:$\lambda_0$}] {};
\draw (0.375,0) -- + (0, -0.02) node[below, yshift = -0.1cm, label={[scale=0.8]center:$\lambda_0 + \frac{\epsilon}{4}$}] {};
\draw (0.625,0) -- + (0, -0.02) node[below, yshift = -0.1cm, label={[scale=0.8]center:$\lambda_0 + \frac{3\epsilon}{4}$}] {};
\draw (0.75,0) -- + (0, -0.02) node[below left, yshift = -0.1cm, label={[scale=0.8]center:$\lambda_0 + \epsilon$}] {};
\draw[thick, red] (0.35,0) -- + (0, +0.02) node[above, yshift = +0.2cm, label={[black, scale=0.8]center:$\theta_0$}] {};
\draw[thick, red] (0.35 - 0.25, 0) -- (0.35 + 0.25, 0); 
\draw[thick, blue] (0.8,0) -- + (0, +0.02) node[above, yshift = +0.2cm, label={[black, scale=0.8]center:$\lambda_j$}] {};
\draw[thick, blue] (0.8 - 0.125, 0) -- (0.8 + 0.125, 0);
\end{tikzpicture}
\caption{On this number line, we illustrate several quantities relating to our algorithm for preparing $\rho$. The red interval represents $[\theta_0-\epsilon/2, \theta_0+\epsilon/2]$. Notice that it completely covers the interval $[\lambda_0 - \epsilon/4, \lambda_0 + \epsilon/4]$. Because $\mathcal{A}$ is run with precision $\epsilon/4$, the proportion of $\ket{\Tilde{\lambda_0}}$ lying in the red interval is at least $\frac{1-\delta}{N}$, and so $\ket{\psi_0}$ will be amplified . We also illustrate some other eigenvalue $\lambda_j$ with $j\neq 0$, and in blue, the interval $[\lambda_j - \epsilon/4, \lambda_j + \epsilon/4]$, in which most of $\ket{\Tilde{\lambda_j}}$ lies. Because $\lambda_j > \lambda_0 + \epsilon$, the blue and red intervals are disjoint, which implies that $\ket{\psi_j}$ will not be amplified by very much.}
\label{fig:QAA_numberline}
\end{figure}
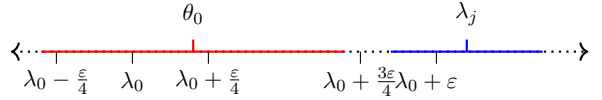

\section{Implementation}\label{app:implementation}
We have implemented a simulation of the algorithms of \Cref{sec:alg} using Qiskit \cite{Qiskit}. The code implementation is available at \url{https://www.github.com/softwareQinc/EigenvalueFinding}~\cite{Kerzner_EigenvalueFinding_2023}. The code uses several simplifications: (1) We prepare the state $\frac{1}{\sqrt{N}}\sum_j\ket{\psi_j}$ directly, instead of preparing $\frac{1}{\sqrt{N}}\sum_j\ket{\psi_j}\ket{\bar{\psi_j}}$, which allows us to use $n$ fewer qubits; (2) for simplicity, we do not use the median trick; (3) we use iterative amplitude estimation \cite{grinko2021iterative} instead of the original {\QAE} algorithm; (4) when rescaling the matrix involved in \Cref{sec:Maxwell_waves}, we use knowledge of $\lambda_0$ and the matrix norm in order to use fewer qubits; and (5) when performing {\QAA}, instead of the usual approach \cite{Brassard_2002} of trying various different powers of the Grover iterate until we obtain the right measurement outcome, we simply perform postselection.

\end{document}